\documentclass[10pt,twocolumn]{IEEEtran}
\usepackage{epsfig,latexsym}
\usepackage{float}
 \usepackage{enumerate}
\usepackage{indentfirst}
\usepackage{amsmath}
\usepackage{amssymb}
\usepackage{times}
\usepackage{subfigure}
\usepackage{cite}
\usepackage{url}
\usepackage{color}
\usepackage{setspace}

\definecolor{green}{rgb}{0.796,0.948,0.816}

\newtheorem{Proposition}{Proposition}

\newtheorem{Lemma}{Lemma}

\newtheorem{Remark}{Remark}

\newtheorem{proof}{proof}
\newtheorem{theorem}{$\mathbf{Theorem}$}

\begin{document}


\title{Optimal Power Allocation { Scheme} for Non-Orthogonal Multiple Access\\ with $\alpha$-Fairness}
\author{Peng Xu and Kanapathippillai Cumanan, \IEEEmembership{Member, IEEE}
\thanks{Peng Xu is with the Chongqing Key Laboratory of Mobile Communications
Technology, Chongqing University of Posts and Telecommunications (CQUPT),
 Chongqing, 400065, P. R. China.
Kanapathippillai Cumanan is with Department of Electronic Engineering,
University of York,
York, UK, YO10 5DD.

The work of Peng Xu was supported in part by the Scientific and Technological Research Program of
Chongqing Municipal
Education Commission under Grant KJ1704088, and in part by  Doctoral  Initial Funding
of Chongqing University of Posts and Telecommunications under Grant A2016-84.
The work of K. Cumanan  was supported by H2020- MSCARISE-2015 under grant number 690750.
} \vspace{-2em}}\maketitle
\begin{abstract}
This paper investigates the optimal power allocation scheme
 for  sum throughput maximization  of non-orthogonal multiple access (NOMA) system
  with $\alpha$-fairness. In contrast to the
 existing fairness NOMA models,
 $\alpha$-fairness can only utilize a single scalar to achieve
 different  user fairness levels. Two different channel state information at the transmitter
 (CSIT) assumptions are considered, namely, statistical and perfect CSIT.
For statistical CSIT, fixed target data rates are predefined,
 and the  power allocation problem is solved for sum throughput
 maximization with $\alpha$-fairness, through characterizing
several properties of the optimal power allocation solution.
For perfect CSIT, {
the optimal power allocation is determined to maximize the instantaneous sum rate
  with $\alpha$-fairness,
 where user rates are adapted according to the instantaneous channel state information (CSI)}.
In particular, a simple alternate optimization (AO) algorithm is proposed,
 which is demonstrated to yield the optimal solution.
 Numerical results reveal that,
 at the same fairness level,
  NOMA  significantly outperforms the conventional orthogonal  multiple access (MA)  for both the scenarios with statistical
  and  perfect CSIT.
  \end{abstract}
\begin{IEEEkeywords}
  Non-orthogonal multiple access, $\alpha$-fairness, outage probability, ergodic rate, power allocation.
\end{IEEEkeywords}

\section{Introduction}\label{i}
Non-orthogonal multiple access (NOMA) enables to realize a balanced tradeoff between spectral efficiency and user fairness, which has been recognized
 as a promising   multiple access (MA) technique  for  future fifth generation (5G) networks
  \cite{saito2013system,li20145g,Ding2016Survey,
  Dai2015Non,Wang2016Dynamic,Ding2014Impact,Ding2016MIMO,Ding2016General_MIMO,ding2016design,Cuma_add,
 ding2014performance,liu2015cooperative,Ding2015cooperative_NOMA2,yang2016outage,
 timotheou2015fairness,Xu2016On,liu2016fairness,shi2015outage}.
 In contrast to  the conventional  MA
(e.g., time-division multiple access (TDMA), etc.), NOMA exploits power domain
to simultaneously serves multiple users at different power levels,
where power allocation at the base station  plays a key role in determining the overall performance of the system.
 Downlink NOMA combines  superposition coding at the transmitter and
 successive interference cancellation (SIC) decoding at each  receiver,
 which can be considered as
 a special case of the conventional broadcast channel (BC) \cite{xu2015new}.
 To maintain user fairness, NOMA always allocates more power to the users with weaker channel gains.

Based on the superposition coding, the works in  \cite{cover1972broadcast}
and \cite{bergmans1974simple} explored
  the capacity region of the degraded discrete
memoryless BC and the Gaussian BC with single-antenna
terminals, respectively.
{ On the other hand,  the ergodic capacity and the outage
  capacity/probability of the fading BC
with perfect channel state information
 at the transmitter (CSIT) were established  in \cite{li2001capacity}
 and \cite{li2001capacity_outage}, respectively.}
 For the concept of ergodic capacity, user rates can be adapted  according to the
 instantaneous
 channel state information (CSI);
 while the concept of outage   is more appropriate for
  applications with stringent delay constraints as  a   predefined rate is assumed for each transmission.
 In \cite{zhang2009downlink}, the  performance of outage capacity was analyzed
without CSIT. However, these works for conventional BCs   have not taken into account
the issue of user fairness, which is different from NOMA with fairness constraints.

Recently, the issue of user fairness has received considerable attention in
a series of NOMA systems  \cite{ding2014performance,
 liu2015cooperative,Ding2015cooperative_NOMA2,yang2016outage,
 timotheou2015fairness,Xu2016On,liu2016fairness,shi2015outage}.
{ The works in \cite{ding2014performance,
 liu2015cooperative,Ding2015cooperative_NOMA2,yang2016outage,Cuma_Sig_Process_Lett_J09,Rahul_J11,Cuma_ICC11,Cuma_VTC09}
 adopted fixed power allocation approaches to guarantee user fairness,
 which can only ensure that the users with weaker channel gains are allocated
with more power and might suffer from poor
user fairness when some users have very poor channel conditions.}
In order to enhance user fairness,
an appropriate power allocation should be adopted at the base station
  for each user message in the superposition coding, similar to the works in
  \cite{timotheou2015fairness,Xu2016On,liu2016fairness,shi2015outage}. In \cite{timotheou2015fairness},
  the max-min and min-max power allocation schemes are proposed to maximize the
  ergodic rate and minimize the outage probability, respectively, {  whereas
  the common outage probability of NOMA with
   one-bit  feedback is minimized in \cite{Xu2016On}}.  { A throughput maximization scheme for a
   multiple-input multiple-output (MIMO) NOMA system is presented in \cite{liu2016fairness}
    by solving the max-min fairness problem}.
   { However, the schemes proposed in \cite{timotheou2015fairness,Xu2016On,liu2016fairness}
   can only achieve absolute fairness\footnote{ The term ``absolute fairness'' means that all users have
     the same performance (e.g., the same outage probability or the same ergodic rate).}, where  the system
   throughput is limited  by the user with the worst channel gain.}
      In \cite{shi2015outage}, the power allocation  approach has been proposed { to maximize}
       the minimum weighted success probability, where  a weighting  vector is exploited to adjust  fairness levels.
However,   { the design of the optimal weighting
vector is a challenging issue, which has not been addressed in \cite{shi2015outage}}. Most recently,  { a proportional fairness
   over a time-domain window size has been presented for NOMA in  \cite{liu2016pro}}.

The main objective of this paper is to investigate the optimal power allocation scheme
 for  sum throughput maximization  of the NOMA system with $\alpha$-fairness constraints.
    {  In existing fairness models in \cite{timotheou2015fairness,Xu2016On,liu2016fairness},
     only absolute fairness can be achieved; while in \cite{shi2015outage},
  a weighting
vector are exploited to adjust the fairness level. However,
 $\alpha$-fairness only utilizes a single scalar, denoted as $\alpha$, to achieve
 different  user fairness levels and well-known
efficiency-fairness tradeoffs \cite{shi2014fairness}.}
 The concept of $\alpha$-fairness was first introduced in \cite{JMo2000Fair} for a fair end-to-end congestion control,
 which generalizes proportional  and  max-min fairness approaches.
 Since then, $\alpha$-fairness
 has been widely incorporated  in  a series of fairness optimization models for
 resource allocation and congestion control (e.g., \cite{Lan2010An,uchida2011an1,song2015energy}).
    More details on fairness in wireless networks can be found in \cite{shi2014fairness}
and the references therein. In general, increasing $\alpha$
results in higher user fairness \cite{Lan2010An}.
  For instance,
 maximum efficiency can be achieved by setting $\alpha=0$,  {   whereas proportional and max-min fairness
 can be achieved by setting $\alpha=1$ \cite{kelly1998rate,nguyen2006a},
 and  $\alpha\rightarrow \infty$  \cite{JMo2000Fair}, respectively.}

In this paper, a downlink NOMA system with two different CSIT assumptions are considered:
statistical and perfect CSIT.
For statistical CSIT, fixed target data rates should be predefined
 for all users, for which, we first analyze the outage probability of each user,
and then formulate the  power allocation optimization framework for sum throughput
 maximization with $\alpha$-fairness. However, this optimization problem is not convex in nature
 due to the non-convex objective function.
  To circumvent this non-convex issue,
 we reformulate the original problem
into an equivalent problem with a simple expression.
Analysis reveals that the equivalent  transformed  problem is convex for the case of $\alpha\geq 1$
and still non-convex for the case of $\alpha<1$. However, for
the case $\alpha<1$, the structure of the optimal solution is characterized
based on
some properties of the optimal power allocation solution,  which demonstrates
that the problem turns out to be convex if we fix the first power parameter and the number of
power parameters that are below $(1-\alpha)/{2}$.

For perfect CSIT, {
the  power allocation  problem is formulated to maximize the instantaneous sum rate
 with $\alpha$-fairness, where user rates are adapted according
 to instantaneous CSI}. We first transform this optimization
   problem into an equivalent problem  by setting a series of
 parameters to denote the sum power allocated to a group of users. Then,
we demonstrate that there exists only one solution to satisfy the Karush-Kuhn
Tucker (KKT) conditions. Furthermore,
a simple alternate optimization (AO) algorithm is proposed to yield  the optimal solution
through solving KKT conditions.
{ The algorithm is developed based on the idea of  AO approach, where each
  KKT condition is solved  by fixing the
 other corresponding parameters. In addition, it is shown that each variable is monotonically increasing in each iteration of the algorithm and therefore it converges.}

 Numerical results reveal that parameter $\alpha$ can adjust the fairness level in terms of
 fairness index \cite{Jain1984@FI}  for both NOMA and TDMA. In addition,
 for the same required fairness index ,
  NOMA  outperforms TDMA in terms of both the sum throughput with statistical CSIT
  and ergodic sum rate with perfect CSIT.
  Moreover, the proposed algorithm for ergodic rate maximization
 converges with less number of iterations than the conventional interior point algorithm in most scenarios.

Throughout this paper,  $\mathbb{P}(\cdot)$ and $\mathbb{E}(\cdot)$
are used  to denote the probability of an event and the expectation of a random variable.
Moreover, $[1:K]$ represents the set $\{1,\cdots,K\}$,
and $\{x_i\}$ indicates the sequence formed by all the possible $x_i$'s.
Furthermore, $\log(\cdot)$ and
 $\ln(\cdot)$ stands for the logarithm with base 2 and
the natural logarithm, whereas $\exp(\cdot)$ denotes the exponential function.



\section{System Model and Problem Formulations}\label{ii}
A downlink NOMA system is considered  with one single-antenna base station and $K$ single-antenna users.
For this network setup, quasi-static block fading is assumed,
where the  channel gains from the base station to all users are  constant during one fading block, but change independently from one fading block to  the next fading block.
The  base station transmits $K$ messages to the users using the NOMA scheme, i.e., it sends
a superposition codeword $x=\sum_{k=1}^K \sqrt{\tilde{P}_k}s_k$ during each fading block, where $s_k$ is
the  signal intended for   user $k$ with $E[|s_k|^2]=1$ and $\tilde{P}_k$ is the power allocated to
 user $k$, which satisfies $\sum_{k=1}^K \tilde{P}_k \leq P$.
 The received signal at user $k$ can be expressed as
{ \begin{align}\label{superposition}
  y_k=h_k\sum_{i=1}^K \sqrt{\tilde{P}_i}s_i+n_k, \ k\in[1:K],
\end{align}}
where the noise  $n_k$ at   user $k$  is assumed to be an additive white Gaussian noise with zero
mean and { unit variance}, and  $h_k$ denotes the channel gain from the base station to
 user $k$. Specifically, $h_k=d_k^{-{\beta}/{2}}g_k$, where  $g_k$ is a normalized
 Rayleigh fading channel
gain with unit variance, $d_k$ is the distance between the
base station and user $k$, and $\beta$ is the path loss exponent. 
  Without loss of generality, it is assumed that $d_1>d_2>\cdots>d_K$.
 In addition, it is also assumed that noises and channel gains associated with all users are mutually independent from each other. In this paper, we consider the case where each superposition codeword spans only a
single fading block.

The users  employ SIC to decode their messages, where the { user order (or
equivalently, decoding order)} is determined by the base station
according to the CSIT assumption discussed later in this section.
It can be assumed without loss of generality that user $k$ is allocated with index $k$.
In the SIC process,  user $k$ will sequentially  decode the messages
 of   users $l $, $l\in[1:k]$ and then successively remove these messages from its  received signal.
When user $k$ decodes the message of user $l$, the signal-to-interference-plus-noise ratio (SINR) can be written as
 \begin{align}\gamma_{l}^{(k)}=\frac{\tilde{P}_lH_k}{H_k\sum_{m=l+1}^K\tilde{P}_m+1},
 \ l\in[1:k],\end{align}
 where we define $H_k=|h_k|^2$, $\forall k\in[1:K]$, for simplicity; obviously, $H_k$  follows an exponential distribution with a mean
$d_k^{-{\beta}}$.

 Next, we will investigate optimal power allocation   from a fairness
perspective, under two main CSI assumptions of  statistical and perfect CSIT.
To model fairness, we adopt $\alpha$-fair utility function \cite{Jain1984@FI}
\begin{align}\label{u(F_k)}
  u_{\alpha}(x)\triangleq\left\{\begin{array}{ll}
    \ln(x), &\textrm{if } \alpha=1,\\
     \frac{x^{1-\alpha}}{1-\alpha}, &\textrm{if } \alpha\neq 1,\alpha\geq 0.
  \end{array}\right., \ x>0,
\end{align}
where $x$ could be throughput or instantaneous rate shown later in this section,
 and different values of $\alpha$ represents different  fairness levels.
Note that the choices of $\alpha=0$ and $\alpha\rightarrow \infty$ represent
 no fairness  and absolute  fairness requirements, respectively.

 \subsection{NOMA with Statistical CSIT}
{ For the  statistical CSIT scenario,  only  statistics of fading channels
    (including channel distributions, means and variances) are available at the transmitter,
and hence fixed target data rates should be predefined
 for all users.   The overhead cost in this scenario would be low as
 the variation of channel statistics is   much more slower than  that of instantaneous CSI.
 Moreover, the user order is determined based on the distance from the base station to each user, where
a user with a larger distance is assigned  with a smaller order index.
Since it is assumed  that $d_1>d_2\cdots>d_K$ previously in this section,
 user $k$ is always allocated with order index $k$. }
 Assume that the base station transmits one message to each user in each block  with
the same fixed target rate $r_0$ bits per channel use (BPCU). \footnote{ Note that  setting different fixed target rates for different users can improve the sum throughput. However,
user fairness will be affected by such a different-rate scheme.
For instance, absolute fairness is
difficult to be achieved if the  data rates of the users are not the same.
 Motivated by this,  the same fixed rate $r_0$  for each user is assumed in this paper.} For this transmission scenario, the outage probability needs  to be evaluated, and the  outage probability for user $k$ can
be expressed as
{ \begin{align}\label{Pr_k}
  &\quad\mathbb{P}_k=\mathbb{P}\left\{\gamma_{l}^{(k)}<\hat{r}_0,\textrm{ for some }
  l\in[1:k]\right\}\nonumber\\
  &=\mathbb{P}\left\{H_k<\max\left\{\frac{\hat{r}_0}{\hat{P}_1},
  \cdots,\frac{\hat{r}_0}{\hat{P}_k}\right\}\right\}\nonumber\\
  &=1-\exp\left(-\max\left\{\frac{\hat{r}_0d_k^{\beta}}{\hat{P}_1},\cdots,
  \frac{\hat{r}_0d_k^{\beta}}{\hat{P}_k}
  \right\}\right),
\end{align}}
where $\hat{r}_0\triangleq 2^{{r}_0}-1$, $\hat{P}_k\triangleq \tilde{P}_k-\hat{r}_0\sum_{m=k+1}^K\tilde{P}_m$ which
can be considered as an {\em equivalent} power for user $k$.
Note that in \eqref{Pr_k}, it is implicitly assumed that
  \begin{align}\label{noma_constraint}
    {\tilde{P}_{k}\geq \hat{r}_0\sum_{m=k+1}^K\tilde{P}_{m}},\ \forall k\in[1:K-1].
  \end{align}
  This power constraint is widely incorporated in general for   NOMA systems as in \cite{ding2014performance,Ding2015cooperative_NOMA2,timotheou2015fairness,shi2015outage}, where
  more power is allocated to a user  with weak channel gains  to guarantee user fairness.

The power constraint can be rewritten as \cite{timotheou2015fairness,Xu2016On}
\begin{align}
  \sum_{k=1}^K (\hat{r}_0+1)^{k-1}\hat{r}_0 \hat{P}_k\leq P. \label{power_constraint}
\end{align}

Furthermore, the throughput of user $k$ is denoted as
\begin{align}\label{F_k}F_k(\{\hat{P}_k\})&\triangleq r_0(1-\mathbb{P}_k)\nonumber\\
&=r_0\exp\left(-\max\left\{\frac{\hat{r}_0d_k^{\beta}}{\hat{P}_1},
\cdots,\frac{\hat{r}_0d_k^{\beta}}{\hat{P}_k}\right\}\right).\end{align}
 To investigate the sum throughput maximization
 with $\alpha$-fairness, we formulate the following optimization problem:
   \begin{subequations}\label{new_problem_long1}
\begin{align}
 \textrm{(F.P1) }&\max_{\{\hat{P}_k\}} \sum_{k=1}^Ku_{\alpha}\left(F_k(\{\hat{P}_k\})\right)
 \label{obj1}\\
   \textrm{s.t. }&
  \eqref{power_constraint},\ \hat{P}_k\geq 0, \ k\in[1:K].\label{constraint11}
\end{align}
\end{subequations}

\subsection{NOMA with Perfect CSIT}
{ In the scenario of  perfect CSIT  in each block,}
user's data rates can be adapted according to the channel conditions without any outage.
However,  the base station needs  to
estimate each channel gain based on pilot symbols transmitted by the users,
which is different from the scenario of statistical CSIT assumption in the previous subsection.
The user order is determined based on instantaneous CSI at the beginning of each fading block.
It is  assumed without loss of generality that $H_1\leq H_2\leq \cdots \leq H_K$.
The  instantaneous  rate for user $k$ can be expressed as \cite{bergmans1974simple}
\begin{align}\label{ratek}
  R_k(\{\tilde{P}_i\})=\ln \left(\frac{1+H_k\sum_{i=k}^K\tilde{P}_i}{1+H_k\sum_{i=k+1}^K\tilde{P}_i}\right),\
  k\in[1:K],
\end{align}
where the rate is measured in nats per channel user (NPCU). Note that  NPCU has been adopted here for
mathematical brevity, however, it can be easily converted  into
 BPCU. The ergodic sum rate can be expressed as $\mathbb{E}\left[\sum_{k=1}^KR_k\right]$.

 To determine the optimal power allocation to maximize the   instantaneous sum rate
 with  $\alpha$-fairness, we formulate the following optimization problem:
   \begin{subequations}\label{problem_rate1}
\begin{align}
 \textrm{(R.P1) }&\max_{\{\tilde{P}_i\}} \sum_{k=1}^K u_{\alpha}(R_k(\{\tilde{P}_i\}))
 \label{objR1}\\
   \textrm{s.t. }&\sum_{i=1}\tilde{P}_i\leq P,\label{constraintR1}
  \\& \tilde{P}_i\geq 0, \ i\in[1:K].
\end{align}
\end{subequations}
%

{ \begin{Remark}Although the Rayleigh fading channel model
 is considered in this paper,
  the formed optimization problems can be easily extended
to more practical channel models, such as the widely used Saleh-Valenzuela multi-path model
\cite{Gao2016Energy}. In particular,  the extensions of problem (R.P1)  to the other
channel models are straightforward as the instantaneous  rates in \eqref{ratek}
are also valid for any other  channel distributions; whereas the outage probabilities in \eqref{Pr_k}
as well as problem (F.P1)  should be modified according to the  channel distribution.
The study of the other multi-path channel models is out of the scope of this paper.
\end{Remark}}

{ \begin{Remark}Note that  the  power allocation problems with
 $\alpha$-fairness will be more complicated  for the scenario
of  multiple antennas at the base station, where
  the optimal user ordering scheme in  the SIC process is still an open problem for MIMO-NOMA \cite{Ding2016MIMO,Ding2016General_MIMO,ding2016design}.
 Thus, for MIMO-NOMA with $\alpha$-fairness,   a possible solution approach is to
   utilize the sub-optimal user ordering schemes in
 \cite{Ding2016MIMO,Ding2016General_MIMO,ding2016design},
and then form  the precoding optimization problems at the base station.
More details of  MIMO-NOMA with $\alpha$-fairness are out of the scope of this paper,
which would be an interesting future direction.
\end{Remark}}

\section{Optimal Power Allocation with \\Statistical CSIT}\label{iii}
In this section, we solve problem (F.P1) to obtain the optimal power allocation scheme for
sum throughput maximization with $\alpha$-fairness.

\subsection{Problem Transformation}
In this subsection, we first convert the problem (F.P1) into a more simple tractable optimization framework. { As the first step in this transformation,
we can prove the following inequality condition on the optimal power allocation of problem (F.P1)\cite{timotheou2015fairness,Xu2016On}:}
\begin{align}\label{constraint12}
  \hat{P}_1\geq \hat{P}_2\geq\cdots\geq \hat{P}_K.
\end{align}
The details of the proof  are omitted here for simplicity.
In addition,  $F_k$ in  \eqref{F_k} can be simplified as
\begin{align}\label{F_k2}
   F_k(\hat{P}_k)=r_0\exp\left(-\frac{\hat{r}_0d_k^{\beta}}{\hat{P}_k}
  \right).
\end{align}

By denoting $P_k\triangleq \hat{P}_k/(\hat{r}_0d_k^{{\beta}})$,
$F_k$ in \eqref{F_k2} can be represented as
\begin{align}\label{F_k3}
   F_k({P}_k)=r_0\exp\left(-\frac{1}{{P}_k}
  \right).
\end{align}
On the other hand, the constraints in \eqref{power_constraint} and \eqref{constraint12} can be rewritten as
\begin{align}
   &\hat{r}_0^2\sum_{k=1}^K \Gamma_k{P}_k\leq P, \textrm{ where }\Gamma_k\triangleq
   (\hat{r}_0+1)^{k-1}d_k^{{\beta}}, \label{power_constraint2}\\
   &d_1^{{\beta}} {P}_1\geq d_2^{{\beta}}{P}_2\geq\cdots\geq d_K^{{\beta}}{P}_K,\label{power_inequality}
\end{align}
respectively. Now, problem (F.P1) can be reformulated as
   \begin{subequations}\label{problem2}
\begin{align}
 \textrm{(F.P2) }&\max_{\{{P}_k\}} \sum_{k=1}^Ku_{\alpha}(F_k(P_k))
 \label{obj2}\\
   \textrm{s.t. }&
  \eqref{power_constraint2} \textrm{ and } \eqref{power_inequality},\
   {P}_k\geq 0, \ k\in[1:K].\label{constraint21}
\end{align}
\end{subequations}

\subsection{Optimal Power Allocation}
In this subsection, we  solve the power allocation  problem (F.P2) for
different cases with the corresponding values of $\alpha$. Note that it is  assumed that the distances of the users are significantly different from each other, such that\footnote{
  The assumption on the distances
 in \eqref{assumption} is reasonable in practical NOMA systems.
For example, multiuser superposition transmission (MUST), a
downlink two-user version of NOMA, has been included in 3rd
generation partnership project long-term evolution advanced
(3GPP-LTE-A) networks \cite{LTE}. For MUST, the base station selects two users, which are far
 from and near to the base station, respectively. Obviously,
  the distance difference
between these two selected users is significantly large.}
{ \begin{align}\label{assumption}
  \frac{d_i^{\beta}}{d_j^{\beta}}>\frac{(\hat{r}_0+1)^{j-1}}{(\hat{r}_0+1)^{i-1}},
  \textrm{ i.e., } \Gamma_i>\Gamma_j,\ \forall i<j.
\end{align}}

\subsubsection{Case $0\leq\alpha<1$}
In this case, based on \eqref{u(F_k)} and \eqref{F_k3},
 problem (F.P2) can be expressed   as
 \begin{subequations}\label{problem3}
\begin{align}
 \textrm{(F.P3) }&f_{(F.P3)}\triangleq\max_{\{{P}_k\}} \quad \sum_{k=1}^K\exp\left(-\frac{1-\alpha}{P_k}\right)
 \label{obj6}\\
   \textrm{s.t. }&
  \eqref{power_constraint2} \textrm{ and } \eqref{power_inequality},\
   {P}_k\geq 0, \ k\in[1:K].\label{constraint6}
\end{align}
\end{subequations}
Problem (F.P3) is challenging to solve due to the non-convex objective function.
To tackle this issue, we first present the following propositions on the objective
 function and the optimal solution.

\begin{Proposition}\label{proposition1}
 When $0\leq \alpha \leq 1$, the function $G(x)\triangleq \exp\left(-\frac{1-\alpha}{x}\right)$ is convex for { $x\in[0,\frac{1-\alpha}{2})$},
  and concave  for { $x\in[\frac{1-\alpha}{2},\infty)$}.
  \begin{proof}
  The second  derivative of $G (x)$ can be derived as
  { \begin{align}G ''(x)=&-\frac{2(1-\alpha)}{x^3}\exp\left(-\frac{1-\alpha}{x}\right)\nonumber\\
 & +\frac{(1-\alpha)^2}{x^4}\exp\left(-\frac{1-\alpha}{x}\right).\end{align}}
   Thus, one can observe that
  $G ''(x)>0$  if { $x\in[0,\frac{1-\alpha}{2})$}, $G ''(x)=0$  if
  { $x=\frac{1-\alpha}{2}$}, and
  $G ''(x)<0$  if { $x\in[\frac{1-\alpha}{2},\infty)$}.
  \end{proof}
\end{Proposition}
\begin{Proposition}\label{proposition2}
 At the optimal solution of problem (F.P3), $P_k^*\geq P_{k-1}^*$, $\forall k\in[2:K]$.
\end{Proposition}
\begin{proof}
  This proposition can be proven by {\em reduction to absurdity}.
Suppose that for the optimal power $\{P_k^*\}$ of problem
(F.P3), there exist $i$ and $j$, $i,j\in[1:K]$,  such that
$i<j$ and $P_i^*>P_j^*$. Now, consider another power pair $(P_i,P_j)\triangleq (P_j^*,P_i^*+\epsilon)$,
where we define \begin{align} \epsilon\triangleq(P_i^*-P_j^*)\left(\frac{\Gamma_i}{
\Gamma_j}-1\right),\end{align} such that
\begin{align}P_i
\Gamma_i+P_j\Gamma_j=
P_i^*\Gamma_i+P_j^*\Gamma_j.\end{align}
From \eqref{assumption}, one can observe that $\epsilon>0$.
Furthermore, it can be obtained that \begin{align}
G (P_i^*)+G (P_j^*)< G (P_i)+G (P_j),\end{align}
since $\epsilon>0$ and
 $G(x)$ is a monotonically increasing function,
which contradicts with the optimality of $(P_i^*,P_j^*)$. This completes the proof of this proposition.
\end{proof}
\begin{Proposition}\label{proposition3}
 For the optimal solution of problem (F.P3), if there are $k_0$ power values, $P_k^*$'s,
 that are below $\frac{\alpha-1}{2}$, then the constraint in \eqref{power_inequality} is
  binding for these power values, i.e., $d_i^{\beta}P_i^*=d_j^{\beta}P_j^*$, $\forall i,j\in[1:k_0]$.
\end{Proposition}
\begin{proof}
  Please refer to Appendix \ref{proof_proposition3}.
\end{proof}

\begin{Remark}\label{remark_structure}
Based on Propositions  \ref{proposition3},  it follows
that the optimal solution of problem (F.P3) should have the
following structure: there are $k_0$ power values, $(P_1,\cdots,P_{k_0})$,
 satisfying $P_k<\frac{\alpha-1}{2}$ and $P_k=\frac{d_1^{\beta}}{d_k^{\beta}}P_1$, $\forall k\in[1:k_0]$,
 and the rest of $(K-k_0)$ power values satisfying $P_{k}\geq\frac{\alpha-1}{2}$, $\forall k\in[k_0+1,K]$. Therefore, the maximum  value of the objective function can be expressed
 as $f^*_{(F.P3)}=\max_{k_0,P_1,\{P_{k_0+1},\cdots,P_K\}} \quad \sum_{k=1}^{k_0} G\left(\frac{d_1^{\beta}}{d_k^{\beta}}P_1\right)+\sum_{k=k_0+1}^K G(P_k)$.
 \end{Remark}

 From problem (F.P3) and Remark \ref{remark_structure}, one can observe that if  $(k_0,P_1)$ is fixed,     the optimal values of $(P_{k_0+1},\cdots,P_K)$ can be obtained
 by solving the following optimization problem:
  \begin{subequations}\label{problem3}
\begin{align}
 \textrm{(F.P4) }&f_{(F.P4)}(k_0,P_1)
 \triangleq\max_{\{P_{k_0+1},\cdots,P_K\}} \quad \sum_{k=k_0+1}^K G(P_k)
 \label{obj7}\\
   \textrm{s.t. }&
  \hat{r}_0^2\sum_{k=k_0+1}^K \Gamma_k P_k\leq P-\hat{r}_0P_1d_1^{\beta}
  ((\hat{r}_0+1)^{k_0}-1),\label{constraint71}\\
  &d_{k_0+1}^{\beta}P_{k_0+1}\geq\cdots\geq d_K^{\beta}P_K,\label{constraint72}\\
   &{P}_k\geq \frac{\alpha-1}{2},\label{constraint73} \ k\in[k_0+1:K].
\end{align}
\end{subequations}
where $(k_0,P_1)\in\mathcal{S}$, and
 $\mathcal{S}$ is defined as
\begin{align}\label{set}
  \mathcal{S}&\triangleq \Big\{(k_0,P_1):\ k_0\in[0:K],\ 0\leq P_1\leq \frac{
  (\alpha-1)d_{k_0}^{\beta}}{2d_{1}^{\beta}},\nonumber\\
  &\hat{r}_0P_1d_1^{\beta}
  ((\hat{r}_0+1)^{k_0}-1)+\frac{\hat{r}_0^2(\alpha-1)}{2}\sum_{k=k_0+1}^K \Gamma_k\leq P\Big\},
\end{align}
such that constraints \eqref{constraint71} and \eqref{constraint73} can be satisfied.

Closed-form
solution to problem (F.P4) is in general not possible. However,
 it can be easily shown  that problem (F.P4) is convex since $G(x)$ is concave when $x\in\left[\frac{1-\alpha}{2},
\infty\right)$ as presented in Proposition 1.
Thus, for a fixed pair $(k_0,P_1)$, problem (F.P4) will be solved later in Section \ref{vi}
   with the help of corresponding numerical solvers.

   { The following work is to find optimal values of $k_0$ and $P_1$,
   denoted as
as $(k_0^*,P_1^*)$,  which can be expressed as
\begin{align}
  (k_0^*,P_1^*)=\arg \max_{(k_0,P_1)\in \mathcal{S}} \sum_{k=1}^{k_0} G\left(\frac{d_1^{\beta}}{d_k^{\beta}}P_1\right) +f_{(F.P4)}^*(k_0,P_1),
\end{align}
where $f_{(F.P4)}^*(k_0,P_1)$ is the maximum  value of the objective function in problem (F.P4)
for a fixed pair $(k_0,P_1)$.
Specifically, in order to find $(k_0^*,P_1^*)$,
a two-dimensional exhaustive search over $k_0$ and $P_1$ should be carried out.
Since $k_0$ is an integer in $[0:K]$
as shown in \eqref{set}, the computational complexity of
this two-dimensional exhaustive search is $O((K+1)\delta)$, where $\delta$
 is the  step size  when searching $P_1^*$ (i.e.,  $\delta$ denotes the searching accuracy
 of $P_1^*$).   }

\subsubsection{Case $\alpha=1$}
In this case, based on \eqref{u(F_k)} and \eqref{F_k3},
 problem (F.P2) can be expressed as
 \begin{subequations}\label{problem3}
\begin{align}
 \textrm{(F.P5) }&\min_{\{{P}_k\}} \quad \sum_{k=1}^K\frac{1}{P_k}
 \label{obj3}\\
   \textrm{s.t. }&
  \eqref{power_constraint2} \textrm{ and } \eqref{power_inequality},\
   {P}_k\geq 0, \ k\in[1:K].\label{constraint31}
\end{align}
\end{subequations}
The following lemma  provides the closed-form expression for the optimal solution of the problem.
\begin{Lemma}\label{lemma1}
  The optimal solution for problem (F.P5) is given by
  \begin{align}\label{solution_(F.P5)}
  P_k&=\frac{1}{\hat{r}_0 \sqrt{\omega\Gamma_k}},
 \textrm{ where }
 \omega=\left(\frac{\hat{r}_0}{P}\sum_{k=1}^K\sqrt{\Gamma_k} \right)^2.
 \end{align}
\end{Lemma}
\begin{proof}
Please refer to Appendix \ref{proof_lemma_1}.
\end{proof}

\subsubsection{Case $\alpha>1$}
In this case, based on \eqref{u(F_k)} and \eqref{F_k3},
 problem (F.P2) can be expressed as
 \begin{subequations}\label{problem3}
\begin{align}
 \textrm{(F.P7) }&\min_{\{{P}_k\}} \quad \sum_{k=1}^K\exp\left(\frac{\alpha-1}{P_k}\right)
 \label{obj5}\\
   \textrm{s.t. }&
  \eqref{power_constraint2} \textrm{ and } \eqref{power_inequality},\
   {P}_k\geq 0, \ k\in[1:K].\label{constraint5}
\end{align}
\end{subequations}
{ The convexity of this problem can be verified through deriving the Hessian matrix of the objective function}.
 Obviously a closed-form expression  for the optimal solution of problem (F.P7) is difficult to obtain,
  however,  this problem will be solved later in Section \ref{vi}
   using corresponding numerical solvers. On the other hand,
   we can verify that, when $\alpha\rightarrow \infty$, absolute user fairness in terms of throughput can be obtained in the following Lemma.
      \begin{Lemma}\label{lemma_absolutefairness}
     When $\alpha\rightarrow \infty$,  $F_i(P_i^*)=F_j(P_j^*)$, $\forall i,j\in[1:K]$,
     where $(P_1^*,\cdots,P_K^*)$ is the optimal solution of problem (F.P7).
   \end{Lemma}
   \begin{proof}
    Please refer to Appendix \ref{proof_lemma_afairness}.
    \end{proof}

\section{Optimal Power Allocation with Perfect CSIT}\label{iv}
In this section, we determine the optimal power allocation to maximize the
instantaneous sum rate  with $\alpha$-fairness by solving problem (R.P1).
\subsection{Problem Transformation}
By denoting $K$ variables as: $b_k\triangleq \sum_{i=k}^K \tilde{P}_i$, $k\in[1:K]$, from \eqref{ratek},
the instantaneous rate of user $k$ can be expressed as
\begin{align}\label{Rkb}
  R_k(b_k,b_{k+1})=\ln \left(\frac{1+H_kb_k}{1+H_kb_{k+1}}\right),\
  k\in[1:K],
\end{align}
 where it is defined $b_{K+1}\triangleq 0$ for the sake of brevity.

 In addition, the power constraint in \eqref{constraintR1} is obviously binding at the optimal solution of problem (R.P1),
i.e.,  $\sum_{i=1}^K\tilde{P}_i= P$ and $b_1=P$. Thus, problem (R.P1) can be reformulated
into the following optimization framework:
 \begin{subequations}\label{problem_rate2}
\begin{align}
 \textrm{(R.P2) }&\max_{\{b_2,\cdots,b_K\}} \sum_{k=1}^K u_{\alpha}(R_k(b_k,b_{k+1}))
 \label{objR2}\\
   \textrm{s.t. }&\label{constraintR2}
   b_k\geq b_{k+1},\ \forall k \in[1:K],\\
   &b_1= P,\ b_{K+1}= 0.
\end{align}
\end{subequations}

The following lemma is required to represent the KKT conditions of problem (R.P2).
\begin{Lemma}\label{lemma_KKT}
  The KKT conditions of problem (R.P2) can be transformed into the following $K$ equations:
  \begin{align}
{f}_{1,k}&(b_k,b_{k+1},b_{k+2})\triangleq\frac{R_{k+1}\left(b_{k+1},b_{k+2}\right)}
    {R_{k}\left(b_{k},b_{k+1}\right)}\nonumber\\
    &-\left(\frac{b_{k+1}+\frac{1}{H_k}}{b_{k+1}+\frac{1}{H_{k+1}}}
    \right)^{1/\alpha}=0,\nonumber\\
     &b_{k+2}< b_{k+1} < b_{k},\ \forall k\in[1:K-1],\label{lemma_bk1_eq2}
\end{align}
\end{Lemma}
\begin{proof}
  Please refer to Appendix \ref{proof_lemma_kkt}.
%
\end{proof}
\begin{Remark}\label{remark_afairness}
  From Lemma \ref{lemma_KKT}, it  can be observed that
 absolute user fairness in terms of instantaneous rate can be obtained when $\alpha\rightarrow \infty$.
 Specifically, $R_{k+1}=R_k$ holds in \eqref{lemma_bk1_eq2},
  $\forall k\in[1:K-1]$, as long as $\alpha\rightarrow \infty$.
\end{Remark}

To obtain the solution through the KKT conditions of problem (R.P2),
 the following theorem is presented. 
\begin{theorem}\label{theorem_unique}
 There is only a unique solution  for the $K-1$ equations in \eqref{lemma_bk1_eq2},
 denoted as $(\hat{b}_2,\cdots,\hat{b}_K)$.
\end{theorem}
\begin{proof}
Please refer to Appendix \ref{proof_theorem_unique}.
\end{proof}
\begin{Remark}
  Theorem \ref{theorem_unique} shows that the KKT conditions of problem (R.P2) are sufficient to
  determine  the optimal solution, i.e., $(\hat{b}_2,\cdots,\hat{b}_K)$ is the optimal solution of
  problem (R.P2). Thus, the conventional interior point algorithm can be utilized
  to solve problem (R.P2). Alternatively, a simple algorithm
  can be developed to solve the $K-1$ equations in \eqref{lemma_bk1_eq2}, as provided in the next subsection.
\end{Remark}

\subsection{Proposed Algorithm}
{  In this subsection, a simple algorithm is developed to solve $K-1$ equations in \eqref{lemma_bk1_eq2}, which yields the optimal solution of the original problem in \eqref{problem_rate2}.}
\begin{Lemma}\label{lemma_bk1}
  For a fixed pair $(b_k,b_{k+2})$, $k\in[1:K-1]$, only a unique $b_{k+1}$ satisfies
   the $k$-th equation in \eqref{lemma_bk1_eq2},  which is the unique root of the following function:
     \begin{align}\label{lemma_bk1_eq}
    \tilde{f}_{1,k}(x)\triangleq\frac{\ln\left(\frac{1+H_{k+1}x}{1+H_{k+1}b_{k+2}}\right)}
    {\ln\left(\frac{1+H_kb_k}{1+H_kx} \right)}-\left(\frac{x+\frac{1}{H_k}}{x+\frac{1}{H_{k+1}}}
    \right)^{1/\alpha},\ b_{k+2}< x < b_{k},
  \end{align}
  where function $\tilde{f}_{1,k}$ is defined
as the same as $f_{1,k}$ in \eqref{lemma_bk1_eq2},
except that $\tilde{f}_{1,k}$ is a single-variable function whereas
${f}_{1,k}$ is a multi-variable function. 
  \end{Lemma}
  \begin{proof}
      We will show that
  function $\tilde{f}_{1,k}(x)$ is monotonically increasing when $b_{k+2}< x < b_{k}$,
  and $\tilde{f}_{1,k}(x)=0$ has only a unique root over $(b_{k+1},b_k)$.
  Specifically,  $\frac{x+\frac{1}{H_k}}{x+\frac{1}{H_{k+1}}}=1+\frac{\frac{1}{H_k}-\frac{1}{H_{k+1}}}{x+\frac{1}{H_{k+1}}}$,
which decreases with $x$ for $x>0$. Recalling \eqref{lemma_bk1_eq}, $\tilde{f}_{1,k}(x)$
 is obviously a monotonically increasing
function when $b_{k+2}\leq x \leq b_{k}$.
Furthermore, $\tilde{f}_{1,k}(x)<0$ as $x\rightarrow b_{k+2}$; $\tilde{f}_{1,k}(x)\rightarrow +\infty$ as $x\rightarrow b_{k}$.
 Therefore, equation $\tilde{f}_{1,k}(x)=0$  has only a unique root, which is denoted as $b_{k+1}^*$.
   Based on the definitions of $\tilde{f}_{1,k}$ and ${f}_{1,k}$,
    $b_{k+1}^*$ is the unique value that satisfies
   the $k$-th equation in \eqref{lemma_bk1_eq2} for a fixed pair $(b_k,b_{k+2})$.
  \end{proof}

\begin{Remark}
As discussed in the proof of Lemma \ref{lemma_bk1}, $\tilde{f}_{1,k}(x)$ is a monotonically increasing function,
hence a simple bisection method can be { utilized} to determine
 the root of equation \eqref{lemma_bk1_eq}, which
is summarized in Algorithm I. 
\end{Remark}


 \begin{table}[tbp]
\hrule
\vspace{1mm}
\noindent {\bf Algorithm I}: Root Search for Fixed $(b_k,b_{k+2})$ in \eqref{lemma_bk1_eq}
\vspace{0mm}
\hrule
\vspace{3mm}
   1: Initialize  $b_{lb}=b_{k+2}$, $b_{ub}=b_k$;

  2: \textbf{while} ($|\tilde{f}_{1,k}(b_{k+1})|>\epsilon_1$) \textbf{do}  

  3: \hspace{0.5cm} Set $b_{k+1}=(b_{lb}+b_{ub})/2$, and calculate $\tilde{f}_{1,k}(b_{k+1})$;

  4: \hspace{0.5cm}  \textbf{if} $\tilde{f}_{1,k}(b_{k+1})>\epsilon_1$ \textbf{then}
   $b_{ub}=b_{k+1}$;

  5: \hspace{0.5cm}  \textbf{else}
  $b_{lb}=b_{k+1}$;

  6: \textbf{until} $|\tilde{f}_{1,k}(b_{k+1})|<\epsilon_1$;

  \vspace{3mm}
\hrule
\end{table}

Motivated by Lemma \ref{lemma_bk1}, a simple AO algorithm is summarized in Algorithm II,
where $b_k^{(t)}$ denotes the value of $b_k$ in the $t$-th iteration.
The basic idea is to alternately solve the $k$-th equation in
\eqref{lemma_bk1_eq2} by fixing the other corresponding variables. Specifically,  in each iteration $t$, the root of the $k$-th equation in
\eqref{lemma_bk1_eq2} is determined  using Algorithm I for a fixed pair $(b_{k}^{(t)},b_{k+2}^{(t-1)})$, $\forall k\in[1:K-1]$. By denoting such a root as $b_{k+1}^*$, the value of $b_{k+1}$
in iteration $t$ is updated
as $b_{k+1}^{(t)}=b_{k+1}^*$, until the required accuracy is achieved. Note that
$\textrm{Norm}\left[\mathbf{f}_1^{(t)}\right]\leq \epsilon_2$
 is utilized as the { stopping criterion}, where
 \begin{align}\label{definition_bf_f1}
 \mathbf{f}_1^{(t)}\triangleq
 \left(f_{1,1}\left(b_1^{(t)},b_2^{(t)},b_3^{(t)}\right),\cdots,f_{1,K-1}
 \left(b_{K-1}^{(t)},b_{K}^{(t)},b_{K+1}^{(t)}\right)\right),
 \end{align}
 and $\textrm{Norm}\left[\cdot\right]$ is the Euclidean distance of a vector.
 In addition, the KKT conditions can be obviously satisfied as
$\textrm{Norm}\left[\mathbf{f}_1^{(t)}\right]\rightarrow 0$, as provided in Lemma \ref{lemma_KKT}.

Next, we analyze the  convergence and optimality of the proposed algorithm.
{ To verify the convergence of the algorithm, the following theorem is required}.

\begin{theorem}\label{theorem_increase}
  For Algorithm II, ${b}_{k}^{(t)}$ is  monotonically increasing with $t$, $\forall k\in[2:K]$.
\end{theorem}
\begin{proof}
Please refer to Appendix \ref{proof_theorem_increase}.
\end{proof}

\begin{Lemma}\label{remark_convergence}
  The proposed AO algorithm in Algorithm II converges.
\end{Lemma}
\begin{proof}
  From Theorem \ref{theorem_increase}, it can be seen  that $b_{k}^{(t)}$ increases with $t$ and its
  upper bound can be defined by $b_1^{(t)}=P$.   Therefore, $\lim_{t\rightarrow \infty}b_{k}^{(t)}$ exists,
  $\forall k\in[2:K]$, and
  the proposed algorithm in  Algorithm II converges.
\end{proof}


 \begin{table}[tbp]
\hrule
\vspace{1mm}
\noindent {\bf Algorithm II}: Proposed Alternate  Algorithm for Problem (R.P2)
\vspace{0mm}
\hrule

\begin{enumerate}
   \item Initialize $t=1$, $b_k^{(0)}=0$, $\forall k\in[2:K]$;
  \item {The $t$-th iteration:}

  Set $k=1$, $b_1^{(t)}=P$, and $b_{K+1}^{(t-1)}=0$;

  \hspace{-0.3cm}\textbf{Repeat} 
 \begin{enumerate} 

   \item Fix $(b_{k}^{(t)},b_{k+2}^{(t-1)})$,
   then find the root of the $k$-th equation in \eqref{lemma_bk1_eq2}, i.e., $b_{k+1}^*$, using Algorithm I.
   \item Set $b^{(t)}_{k+1}=b_{k+1}^*$;
     \item Update $k=k+1$.

    \hspace{-0.7cm}\textbf{Until} $k=K-1$;
   \end{enumerate}

   \item Update $t=t+1$ and repeat Step 2) until $\textrm{Norm}\left[\mathbf{f}_1^{(t)}\right]\leq \epsilon_2$;
\end{enumerate}
\hrule
\end{table}

{  To validate the optimality of Algorithm II, the following lemma is provided.}
\begin{Lemma}\label{remark_optimal}
  The proposed algorithm achieves the optimal solution for problem (R.P2).
\end{Lemma}
\begin{proof}
  Since Algorithm II converges as shown in Theorem \ref{theorem_increase} and Lemma \ref{remark_convergence},
  limit $\bar{b}_k=\lim_{t\rightarrow \infty} b_k^{(t)}$ exists, $k\in[2:K]$, and
    $\tilde{f}_{1,k}(\bar{b}_{k+1})=0$  for the given  pair
   $(\bar{b}_{k},\bar{b}_{k+2})$ in \eqref{lemma_bk1_eq}, $\forall k\in[1:K-1]$. Thus, from
   Lemma \ref{lemma_KKT}, it can be observed that solution $(\bar{b}_2,\cdots,\bar{b}_K)$
   satisfies the KKT conditions of problem (R.P2). Furthermore,
    we know  from Theorem \ref{theorem_unique} that solution $(\bar{b}_2,\cdots,\bar{b}_K)$
  is the unique solution of the KKT functions in \eqref{lemma_bk1_eq2}, i.e., Algorithm II yields the optimal solution for problem (R.P2).
\end{proof}
{ \subsection{Complexity of Algorithm II}
The complexity of Algorithm II is mainly determined by
 two crucial parameters:  the number of
arithmetic operations in each iteration and the speed of convergence.

For  each iteration, the number of arithmetic operations
involved in the proposed algorithm is $O((K-1)\log(1/\epsilon_1))$ since $K-1$ bisection searches are
required with $\epsilon_1$ solution accuracy in Algorithm I. In contrary, the
conventional interior point algorithm requires  $O((K-1)^3)$
arithmetic operations for each iteration \cite{Ye1997Interior},
 which does not have any impact by
 $\epsilon_1$,  however significantly increases with $K$.

 The  convergence speed of Algorithm II is
difficult to estimate due to the very  complicated expression of the functions in \eqref{lemma_bk1_eq2}. {
However, we demonstrate the speed of the convergence with the help of numerical results
later in Section \ref{vi},}
which reveals that the proposed algorithm converges faster than the interior point algorithm in most scenarios.}

\section{Discussion}\label{v}
{ In this section, we discuss an appropriate evaluation criterion of the proposed
$\alpha$-fairness scheme.} The $\alpha$-fairness is a qualitative fairness measure
 of  user throughput or
instantaneous rate \cite{shi2014fairness}. To evaluate
quantitative fairness, there is a widely used measurement, known as ``Jain's Index'' or ``Fairness Index'' (FI),
which is defined as \cite{Jain1984@FI}
\begin{align}
  {\rm FI}(\{x_k\})\triangleq \left(\sum_{k=1}^K x_k\right)^2\big/\left(K\sum_{k=1}^K x_k^2\right),
\end{align}
where $x_k$ could be either $F_k$ or $R_k$, and  FI could take any value over $[1/K,1]$.
A larger
FI generally represents a higher fairness level; the case FI=1 corresponds to absolute fairness.
Moreover, for statistical CSIT, FI turns out to be
 long-term fairness within a large number of blocks;
 while for perfect CSIT,  FI represents short-term fairness within each block.

In general, different values of FI can be achieved by adjusting $\alpha$ \cite{Lan2010An}.
For instance, as shown in Lemma \ref{lemma_absolutefairness}
and Remark \ref{remark_afairness}, $ {\rm FI}(\{F_k\})= {\rm FI}(\{R_k\})=1$ as $\alpha\rightarrow \infty$.
Therefore, we can appropriately choose $\alpha$ to achieve the
fairness index requirement (FIr), where FIr$\in[1/K,1]$.
The corresponding optimization problem can be defined as follows:
 \begin{subequations}\label{problem}
\begin{align}
 &\max_{0\leq\alpha\leq 1} \sum_{k=1}^K x_k^*
 \label{objR2}\\
   \textrm{s.t. }&{\rm FI}(\{x_k^*\})\geq{\rm FIr}, \ \alpha\geq 0,
\end{align}
\end{subequations}
where for a given $\alpha$, $x_k^*=F_k(P_k^*)$ in the case of statistical CSIT  and $x_k^*=R_k(b_k^*,b_{k+1}^*)$
in the case of perfect CSI. Note that  the optimal solutions
of problem (F.P2) and (R.P2) are denoted as $\{P_k^*\}$ and $\{b_k^*\}$, respectively.

Note that increasing $\alpha$
does not necessarily { increase} FI as shown in \cite{sediq2013optimal}.
Thus, in general, a one-dimensional search
 is required to find the optimal value of $\alpha$, denoted as $\alpha^*$,
  for the problem defined in \eqref{problem}.
 However, in most scenarios,  $\sum_{k=1}^K x_k^*$ and FI decreases and increases with $\alpha$, respectively, as shown in many existing works (e.g., \cite{Lan2010An,cheng2008an,song2015energy}). Thus, a simple bisection method
 will be utilized to find $\alpha^*$ in most scenarios later  in Section \ref{vi}.

\vspace{-1em}
\section{Numerical Results}\label{vi}
In this section, computer simulation  results are provided to evaluate the sum throughput
 and the ergodic sum rate of NOMA with $\alpha$-fairness. In these simulations,
 some parameters for the considered NOMA system are set as follows.
{ The small scale
fading gain is Rayleigh distributed, i.e., $g_i\sim \mathcal{CN}(0,1)$.
Furthermore, the noise at each user is assumed to be an additive white Gaussian variable with zero mean
and unit variance. In addition,
 the distance between the base station and user $k$ is defined as $d_k=1.5^{K-k}$,
 and the path loss exponent is chosen as 2 to reflect a favorable propagation condition.}
 \footnote{ The parameter settings $d_k=1.5^{(K-k)}$  and $\beta=2$ show that the assumption
in  (17) is valid even when the distance-ratio $d_{k+1}/d_k$ and the path loss exponent
  $\beta$ are small or moderate.  Due to space limitation, the other choices of parameters have not
  been considered in this paper.}
Since the variance of noise power is unity,   the transmit signal-to-noise-ratio (SNR) is equivalent to the transmit power $P$.
\vspace{-1em}
\subsection{Benchmark Schemes}\label{subsection_benchmark}
 Two benchmark transmission schemes of TDMA and NOMA with fixed power allocation (i.e., fixed NOMA) are considered as explained in the following.
\subsubsection{TDMA  Scheme}
{ The  TDMA transmission method is chosen as one of the benchmark schemes in this evaluation,
 as it is
equivalent to any orthogonal MA scheme \cite[Sec. 6.1.3]{tse2005fundamentals}.}
For TDMA transmission, each fading block is assumed to be equally
 divided into $K$ time slots, where   user
$k$ occupies  the $k$-th time slot. By defining the power allocated to   user $k$ as $P_{k}^T$,
the power constraint for the TDMA scheme can be expressed as $$\frac{1}{K}\sum_{k=1}^KP_{k}^T\leq P.$$

Now, similar to the problems (F.P1) an (R.P1) in
Section \ref{ii}, 
 one can formulate two power allocation problems for the TDMA scheme with statistical and
 perfect CSIT, respectively. Furthermore,
 these two  new TDMA power allocation problems can be solved using similar approaches as in Section \ref{iii} and \ref{iv}. The details of these approaches are omitted here
due to space limitations.

\subsubsection{Fixed NOMA} In order to demonstrate the benefits of power allocation,
  NOMA with fixed power allocation is used as another benchmark scheme.
     In particular,  the NOMA
transmission scheme  in Section \ref{ii} is also utilized, but the power allocation
scheme is fixed as $$\tilde{P}_k=\frac{2^{K-k}P}{2^K-1}, \ k\in[1:K],$$
for both  statistical and
 perfect CSIT. Note that this fixed power allocation scheme is similar to the one
 in \cite{yang2016outage}  with a slight modification.

\begin{figure} \centering
\subfigure[Sum throughput vs $r_0$] { \label{fig:Thr_1a}
\includegraphics[width=0.9\columnwidth]{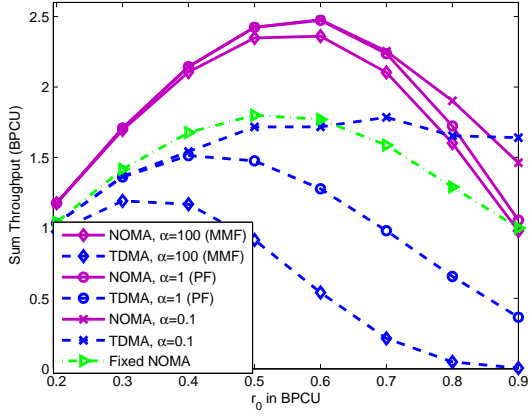}
}
\subfigure[Fairness Index vs $r_0$] { \label{fig:Thr_1b}
\includegraphics[width=0.9\columnwidth]{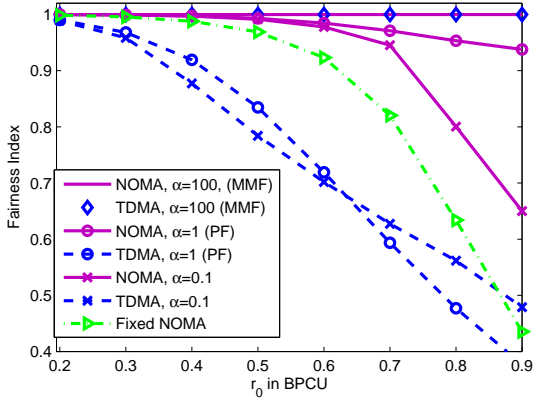}
}
\caption{Sum throughput and fairness index (FI) vs the transmission rate $r_0$ in BPCU, where
 SNR = 20 dB, $K=6$, $\alpha$ = 100, 1, 0.1; MMF and PF denotes max-min fairness  and
   proportional fairness, respectively.}
\label{fig:Thr_1}
\end{figure}
\vspace{-1em}
\subsection{Statistical CSIT}
This subsection focuses on the sum throughput performance of NOMA with $\alpha$-fairness
  and statistical CSIT.
Figs. \ref{fig:Thr_1a} and \ref{fig:Thr_1b} compare the sum throughput and FI
 of NOMA { employing  optimal power allocation  proposed in Section \ref{iii}} with
  the benchmark schemes as a function of the transmission rate $r_0$,   where we set
   $K=6$, SNR = 20 dB, and $\alpha$ = 100, 1, 0.1. 
   \footnote{  Note that max-min fairness (MMF) and
   proportional fairness (PF) can be achieved when
   $\alpha=100$ (i.e., $\alpha$ is sufficiently large) and  $\alpha=1$,
   respectively \cite{shi2014fairness}.}
  As seen in these two sub-figures, NOMA with optimal power allocation enjoys both larger
  sum throughput and FI than the fixed NOMA scheme and the TDMA scheme with optimal power allocation,
  for $\alpha=100$ or 1.
  Moreover, increasing $\alpha$ decreases sum throughput and increases FI for NOMA,
  and absolute fairness can be achieved with $\alpha=100$, which supports the discussions in Lemma \ref{lemma_absolutefairness}.
  { For $\alpha=0.1$, although TDMA with optimal power allocation has a larger sum throughout
   when $r_0=0.9$ BPCU as shown   in Fig. \ref{fig:Thr_1a}, its FI (0.48) is
    lower than the one  achieved by NOMA (0.65).
This is due to the fact that an additional power constraint is imposed on NOMA in \eqref{noma_constraint},
which might reduce the sum throughput with small values of $\alpha$, however, it
 can guarantee the fairness level of NOMA. From Fig. \ref{fig:Thr_1b},
  one can observe that decreasing
$\alpha$ from $1$ to $0.1$ results in the improvement of FI for TDMA when $r_0\geq 0.6$,
 which  is consistent with  the conclusion made in \cite{sediq2013optimal} that increasing $\alpha$ does not necessarily increase FI.}

\begin{figure}[tbp]
    \begin{minipage}[t]{1\linewidth} 
    \centering
    \includegraphics[width=3.1in]{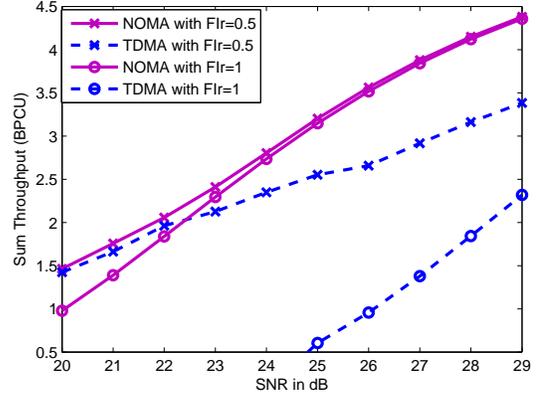}\vspace{-1em}
    \caption{Sum throughput vs SNR, where $r_0=0.9$ BPCU, $K=6$,  FIr = 0.5, 1.}
    \label{fig:Thr_2}
  \end{minipage}%
  \end{figure}
  \begin{figure}[tbp]
  \begin{minipage}[t]{1\linewidth}
    \centering
    \includegraphics[width=3.1in]{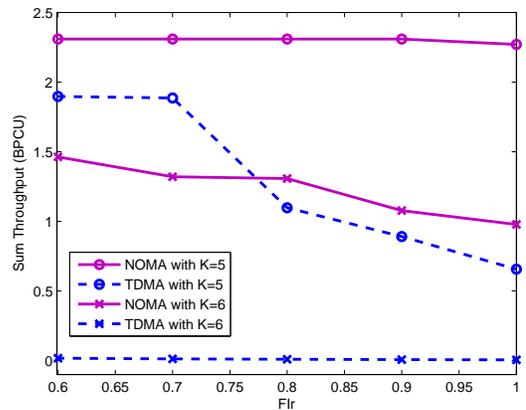}\vspace{-1em}
    \caption{Sum throughput vs the  fairness index requirement (FIr), where $r_0=0.9$ BPCU, SNR
    = 20 dB. $K$ = 5, 6.}
    \label{fig:Thr_3}
  \end{minipage}
\end{figure}

  For a fair comparison between NOMA and TDMA schemes, the same FI is required
  in Figs. \ref{fig:Thr_2} and \ref{fig:Thr_3}. Specifically, we utilize $\alpha$ to
  adjust the value of FI as shown in problem \eqref{problem}, where a bisection search
  is adopted by NOMA, whereas an exhaustive search needs to be adopted by TDMA since
  its FI does not necessarily increase with $\alpha$ as shown in Fig. \ref{fig:Thr_1b}.
  In Fig. \ref{fig:Thr_2}, the sum throughput is depicted as a function of SNR,
  where $r_0=0.9$ BPCU, $K=6$, and the required
  FI is set as FIr = 0.5 or 1. In Fig. \ref{fig:Thr_3},
   the sum throughput is presented as a function of FIr,
  where $r_0=0.9$ BPCU, SNR = 20 dB, $K=5$ or 6.
  From these two figures, one can observe that moderate or high FIr significantly  decreases
  the sum throughput of TDMA, however, it has  less impact on NOMA,
  i.e., NOMA provides a significant performance gain compared to TDMA with moderate or high
  FIr.

\begin{figure} \centering
\subfigure[Ergodic sum rate vs SNR] { \label{fig:rate_1a}
\includegraphics[width=0.9\columnwidth]{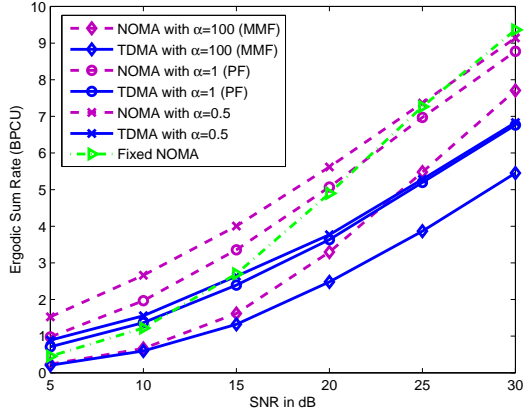}
}
\subfigure[Average fairness Index vs SNR] { \label{fig:rate_1b}
\includegraphics[width=0.9\columnwidth]{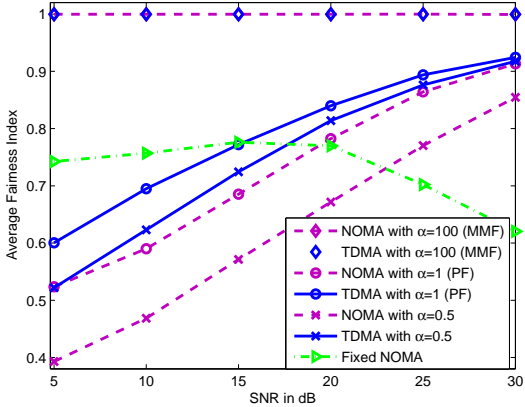}
}
\caption{ Ergodic sum rate and average FI vs SNR in dB, where $K=5$, $\alpha$ = 100, 1, 0.5.}
\label{fig:rate_1}
\end{figure}\vspace{-1em}

\subsection{Perfect CSIT}
This subsection focuses on the ergodic rate performance of NOMA with
$\alpha$-fairness
  and perfect CSIT.
Figs. \ref{fig:rate_1a} and \ref{fig:rate_1b} compare the sum throughput and FI
 of NOMA { employing  optimal AO power allocation algorithm
 proposed in Section \ref{iv}} with
  the benchmark schemes as a function of SNR,   where the parameters are  set as
   $K=5$ and $\alpha$ = 100, 1, 0.5. 
 { As seen in these two sub-figures, the
fixed NOMA has a large ergodic sum rate
  but a very poor average  FI when SNR = 30 dB.}
  On the other hand, NOMA with optimal power allocation has a larger
  ergodic sum rate with a low average  FI  compared to the TDMA scheme.
  For both NOMA and TDMA, increasing $\alpha$ decreases  ergodic sum rate, however
   increases average  FI.
    The absolute fairness can be achieved with $\alpha=100$, which validates  the discussions in
    Remark \ref{remark_afairness}.

\begin{figure}[tbp]
   \begin{minipage}[t]{1\linewidth} 
    \centering
    \includegraphics[width=3.2in]{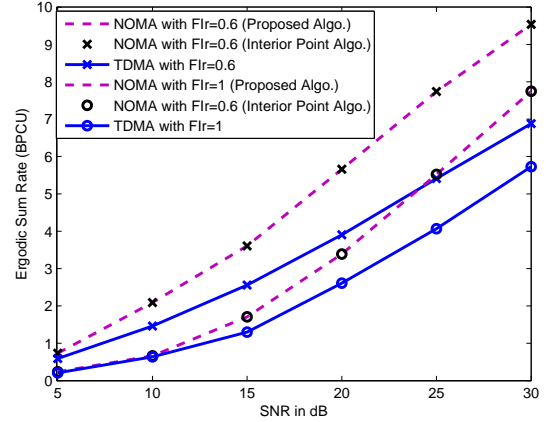}\vspace{-1em}
    \caption{Ergodic sum rate vs SNR in dB, where  $K=5$, SNR = 20 dB, FIr = 0.6, 1.}
    \label{fig:rate_2}
  \end{minipage}%
  \end{figure}
  \begin{figure}[tbp]
  \begin{minipage}[t]{1\linewidth}
    \centering
    \includegraphics[width=3.2in]{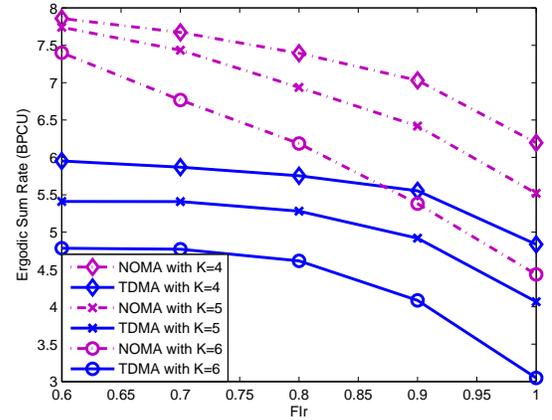}\vspace{-1em}
    \caption{Ergodic sum rate vs FIr,  where SNR = 20 dB, and $K$ = 4, 5, 6.}
    \label{fig:rate_3}
  \end{minipage}
\end{figure}

  In order to make a fair comparison between NOMA and TDMA schemes, the same FI is required
  in Figs. \ref{fig:rate_2} and \ref{fig:rate_3}. Specifically, we utilize $\alpha$ to
  adjust the value of FI as shown in problem \eqref{problem}, where a bisection search
  is adopted by both NOMA and TDMA schemes.
  In Fig. \ref{fig:Thr_2}, the  ergodic sum rate is depicted as a function of SNR,
  where  $K=5$, and  FIr = 0.6 or 1. In Fig. \ref{fig:Thr_3},
   the  ergodic sum rate is depicted as a function of FIr,
  with SNR = 20 dB, $K=4$, 5 or 6.
  As seen in these two figures, one can observe that NOMA
  provides a significant performance gain than the TDMA scheme in terms of  ergodic sum rate
  at the same required fairness level. Moreover, the proposed power allocation algorithm
  achieves the same   ergodic sum rate as the conventional interior point algorithm,
  as shown in Fig. \ref{fig:rate_2}.

  Figs. \ref{fig:Iteration_K4} and \ref{fig:Iteration_K8} compare
  convergence speeds of the proposed algorithm in Section \ref{iv} (i.e., Algorithm II)
  and the conventional interior point algorithm with $K=4$ and 8, respectively.
   Since $\textrm{Norm}\left[\mathbf{f}_1^{(t)}\right]$
 is utilized as the stopping criterion for each fading block, we depict
 its average value as a function of the number of iterations, where
 the required accuracy of Algorithm I (involved in Algorithm II) is set as $\epsilon_1=10^{-5}$, and
 $\alpha=5$,
 2, 1. As evidenced by these two sub-figures, { one can observe} that the proposed algorithm
 converges more faster than the interior point algorithm in most scenarios,
 except with $K=8$, $\alpha=2$ and the number of iterations is larger than 20.

\begin{figure} \centering
\subfigure[$K=4$] { \label{fig:Iteration_K4}
\includegraphics[width=0.85\columnwidth]{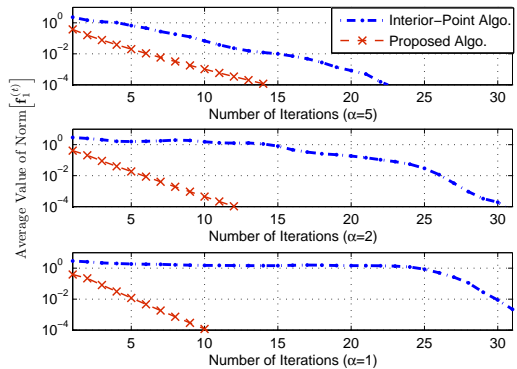}
}
\subfigure[$K=8$] { \label{fig:Iteration_K8}
\includegraphics[width=0.85\columnwidth]{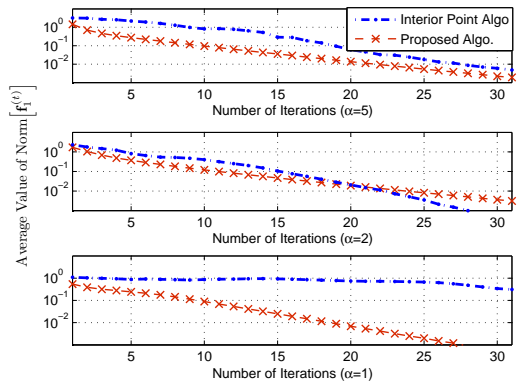}
}
\caption{The average value of Norm$\left[\mathbf{f}_1^{(t)}\right]$ vs the number of iterations, where SNR = 20 dB; $K$ = 4, 8;
$\alpha$ = 5, 2, 1.}
\label{fig:Iteration_K}
\end{figure}
\vspace{-1em}
\section{Conclusions}\label{vii}
{ This paper investigated $\alpha$-fairness based  power allocation schemes
 for  sum throughput and ergodic rate maximization problems in a downlink NOMA system with
statistical and perfect CSIT.}
For statistical CSIT, the outage probability of each user was analyzed,
and the  power allocation strategy was developed
 for sum throughput
 maximization with $\alpha$-fairness. Specifically,
  the original non-convex sum throughput maximization problem was converted into
   an equivalent problem and demonstrated that
the transformed equivalent problem is convex for the case of $\alpha\geq 1$.
In addition, it was shown
that the problem turns out to be convex for $\alpha<1$ by fixing
 the first power parameter and the number of
power parameters that are below $\frac{1-\alpha}{2}$.
{ Next, the instantaneous sum rate
 maximization with $\alpha$-fairness was solved for perfect CSIT, for which it
  was proven that there exists only one solution to satisfy the corresponding KKT conditions. Then,
a simple AO algorithm was developed  to solve these KKT equations.}
As this work only considered single antenna NOMA, an interesting future work  is to extend  to  MIMO NOMA with fairness constrains. Moreover, considering user fairness for the other
 more practical channel model (e.g.,  Saleh-Valenzuela multi-path model \cite{Gao2016Energy}) or
 considering user fairness
over a time-domain window  would be also one of the possible future directions.
\vspace{-1em}
\appendices
\section{Proof of Proposition \ref{proposition3}}\label{proof_proposition3}
 This proposition can be proven by {\em reduction to absurdity}. Denote the optimal power of problem
(F.P3) as $\{P_k^*\}$.
Without loss of generality, it can be assumed that there exist $i<j$, $i,j\in[1:k_0]$, such that
the constraint in \eqref{power_inequality} is not
  binding, i.e.,
$d_i^{\beta}P_i^*>d_j^{\beta}P_j^*$.
From Proposition \ref{proposition2} and the definition of $k_0$,
  $0<P_i^*<P_j^*<\frac{1-\alpha}{2}$ can be obtained.
Now, consider another  power pair $\left(P_i^*-\epsilon_1,
P_j^*+\epsilon_2\right)$, where 
$(\epsilon_1,\epsilon_2)\triangleq\left(\frac{\epsilon}{\Gamma_i},
\frac{\epsilon}{\Gamma_j}\right)$, and $\epsilon$ satisfies
\begin{align}
0<\epsilon< \min\left\{\Gamma_iP_i^*,
\Gamma_j\left(\frac{1-\alpha}{2}-P_j^*\right),
\frac{d_i^{\beta}P_i^*-d_j^{\beta}P_j^*}
{{d_i^{\beta}}/{\Gamma_i}+{d_j^{\beta}}/{\Gamma_j}}\right\}.
\end{align}
Obviously we have  \begin{align}
&0\leq P_i^*-\epsilon_1<P_j^*+\epsilon_2<\frac{1-\alpha}{2}, \label{addeq1}\\
& d_i^{\beta}(P_i^*-\epsilon_1)\geq d_j^{\beta}(
P_j^*+\epsilon_2),\label{addeq2}\\
&\Gamma_i(P_i^*-\epsilon_1)+\Gamma_j(P_j^*+\epsilon_2)=
\Gamma_iP_i^*+\Gamma_jP_j^*,\label{addeq3}
\end{align}
where \eqref{addeq1} implies that the value of $k_0$ will { remain the same} by
replacing the power pair  $(P_i^*,P_j^*)$ by $(P_i^*-\epsilon_1,P_j^*+\epsilon_2)$;
\eqref{addeq2} and \eqref{addeq3} ensure that $(P_i^*-\epsilon_1,P_j^*+\epsilon_2)$ satisfies
 the power constraints in \eqref{power_constraint2} and \eqref{power_inequality}, respectively.
 Next, we need to verify that $G(P_i^*-\epsilon_1)+G(P_i^*+\epsilon_2)>G(P_i^*)+G(P_j^*)$,
 where $G(x)$ is defined in Proposition \ref{proposition1}.



 Based on   {\em Lagrange mean value theorem},  there exists some $\varepsilon_1\in(P_i^*-\epsilon_1,
 P_i^*)$ and  $\varepsilon_2\in(P_j^*,
 P_j^*+\epsilon_2)$ such that
 \begin{align}
  G(P_i^*)- G(P_i^*-\epsilon_1)=\epsilon_1G'(\varepsilon_1),\label{addeq4}\\
   G(P_j^*+\epsilon_2)- G(P_j^*)=\epsilon_2G'(\varepsilon_2).\label{addeq5}
 \end{align}
 Note that since $\varepsilon_1<\varepsilon_2$ and $G''(x)>0$ if $x\in[0,\frac{1-\alpha}{2})$,
  as shown in Proposition \ref{proposition1}, $G'(\varepsilon_1)<G'(\varepsilon_2)$ holds;
  furthermore, since $G'(x)=\frac{1-\alpha}{x^2}\exp\left(-\frac{1-\alpha}{x}\right)$, $x\geq0$,
  $0<G'(\varepsilon_1)<G'(\varepsilon_2)$ can be obtained. In addition,
 $\epsilon_1<\epsilon_2$ holds since { $\Gamma_i>\Gamma_j$} shown in \eqref{assumption}.
 Thus, from \eqref{addeq4} and \eqref{addeq5}, one can observe that $G(P_i^*)- G(P_i^*-\epsilon_1)
 <G(P_j^*+\epsilon_2)- G(P_j^*)$.

In summary, power pair  $
\left(P_i^*-\epsilon_1,
P_j^*+\epsilon_2\right)$  yields a larger  value of the objective function in problem (F.P3), which
 contradicts with the optimality of  $
\left(P_i^*,P_j^*\right)$.
Therefore, the case $d_i^{\beta}P_i^*>d_j^{\beta}P_j^*$ is not optimal,
and $d_i^{\beta}P_i^*=d_j^{\beta}P_j^*$ holds at the optimal solution of problem (F.P3).
\vspace{-1em}
\section{Proof of Lemma \ref{lemma1}}\label{proof_lemma_1}

To solve problem (F.P5), we first consider the following problem by relaxing the constraint in \eqref{power_inequality} of problem (F.P5):
 \begin{subequations}\label{problem31}
\begin{align}
\textrm{(F.P6)}&\min_{\{{P}_k\}} \quad \sum_{k=1}^K\frac{1}{P_k}
 \label{obj31}\\
   \textrm{s.t. }&
  \eqref{power_constraint2},\
   {P}_k\geq 0, \ k\in[1:K].\label{constraint311}
\end{align}
\end{subequations}
The Lagrangian function
 for this problem  is defined as:
\begin{align}
  \mathcal{L}(\{P_k\},\omega,\{\lambda_k\})\triangleq \sum_{k=1}^K
  \frac{1}{P_k}&+\omega\left[\hat{r}_0^2\sum_{k=1}^K \Gamma_k{P}_k-P\right]\nonumber\\
 &-\sum_{k=1}^K \lambda_kP_k,
\end{align}
where $\omega,\lambda_k\geq0$ are Lagrange multipliers. The KKT conditions are
given by
\begin{align}\label{largrange}
  \frac{\partial \mathcal{L}}{\partial P_k}=-\frac{1}{P_k^2}+\hat{r}_0^2\omega \Gamma_k-\lambda_k=0.
\end{align}
In addition, from the complementary slackness conditions (omitted here for simplicity), obviously we have
$\lambda_k=0$ and $\omega>0$, and the power constraint in \eqref{power_constraint2} is binding.
Therefore, from \eqref{power_constraint2} and \eqref{largrange}, the optimal solution of problem (F.P6)
can be obtained as shown in \eqref{solution_(F.P5)}.

From \eqref{solution_(F.P5)}, one can observe that $d_k^{\beta}P_k$ decreases with $k$, which means that the
constraint in \eqref{power_inequality} is satisfied. Thus, problems (F.P5) and (F.P6) have the same optimal solution.
\vspace{-1em}
\section{Proof of Lemma \ref{lemma_absolutefairness}}\label{proof_lemma_afairness}
 Lemma \ref{lemma_absolutefairness}  can also be proven by {\em reduction to absurdity}.
 Denote the optimal power of problem
(F.P7) as $\{P_k^*\}$; based on Proposition \ref{proposition2}, it holds that $P_i^*\leq P_j^*$,
$\forall i<j, \ i,j\in[1:K]$.
Assume without loss of generality that there exist $i$ and $j$ satisfying $i<j$, $i,j\in[1:K]$, such that
$0<P_i^*<P_j^*$. Consider another  power pair $
\left(P_i^*+\epsilon_1,
P_j^*-\epsilon_2\right)$, where $\epsilon_1+\epsilon_2<P_j^*-P_i^*$ and
$(\epsilon_1,\epsilon_2)\triangleq\left(\frac{\epsilon}{\Gamma_i},
\frac{\epsilon}{\Gamma_j}\right)$ for $\epsilon>0$.
Obviously we have \begin{align}&P_i^*+\epsilon_1<
P_j^*-\epsilon_2,\label{addeq6}\\
&d_i^{\beta}(P_i^*+\epsilon_1)\geq d_j^{\beta}(
P_j^*-\epsilon_2),\label{addeq7}\\
& \Gamma_i(P_i^*+\epsilon_1)+\Gamma_j(P_j^*-\epsilon_2)
=\Gamma_iP_i^*+\Gamma_jP_j^*,\label{addeq8}\end{align}
where \eqref{addeq7} and \eqref{addeq8} ensure that $\left(P_i^*+\epsilon_1,
P_j^*-\epsilon_2\right)$ satisfies
 the power constraints in \eqref{power_constraint2} and \eqref{power_inequality}, respectively.
 Denote the function $G_1(x)\triangleq\exp\left(\frac{\alpha-1}{x}\right)$, where $x> 0$,
so the objective function in problem (F.P7) can be expressed as $\sum_{k=1}^K G_1(P_k)$.
Next, we will verify that $G_1(P_i^*+\epsilon_1)+G_1(P_i^*-\epsilon_2)>G_1(P_i^*)+G_1(P_j^*)$.

 Based on   {\em Lagrange mean value theorem},  there exists some $\varepsilon_1\in(P_i^*,
 P_i^*+\epsilon_1)$ and  $\varepsilon_2\in(P_j^*-\epsilon_2,
 P_j^*)$ such that
 \begin{align}G_1(P_i^*+\epsilon_1)-G_1(P_i^*)=\epsilon_1 G_1'(\varepsilon_1),\label{G11}\\
 G_1(P_j^*)-G_1(P_i^*-\epsilon_2)=\epsilon_2 G_1'(\varepsilon_2).\label{G12}\end{align}
 Since the derivative of $G_1(x)$ is
\begin{align}G_1'(x)&=-\frac{\alpha-1}{x^2}\exp\left(\frac{\alpha-1}{x}\right)\nonumber\\
&=-(\alpha-1)\exp\left(-2\ln(x)+\frac{\alpha-1}{x}\right),\nonumber\end{align} we have
\begin{align}
  \frac{\epsilon_1 G_1'(\varepsilon_1)}{\alpha-1}=
  -\exp\left(\ln(\epsilon_1)-2\ln(\varepsilon_1)+\frac{\alpha-1}{\varepsilon_1}\right),\label{G1'(1)}\\
   \frac{\epsilon_2 G_1'(\varepsilon_2)}{\alpha-1}=
  -\exp\left(\ln(\epsilon_2)-2\ln(\varepsilon_2)+\frac{\alpha-1}{\varepsilon_2}\right).\label{G1'(2)}
\end{align}
Furthermore, from \eqref{addeq6}, one can easily  obtain that $\varepsilon_1<\varepsilon_2$. Thus, from
\eqref{G1'(1)} and \eqref{G1'(2)}, we have
\begin{align}\epsilon_1 G_1'(\varepsilon_1)<\epsilon_2 G_2'(\varepsilon_2)<0
\textrm{ as }\alpha\rightarrow \infty.\label{G1'}\end{align}

Now, combing \eqref{G11}, \eqref{G12} with \eqref{G1'}, $G_1(P_i^*+\epsilon_1)+G_1(P_i^*-\epsilon_2)>G_1(P_i^*)+G_1(P_j^*)$ holds when $\alpha\rightarrow \infty$.
In summary, power pair  $
\left(P_i^*+\epsilon_1,
P_j^*-\epsilon_2\right)$  yields a smaller value of the objective function  for problem (F.P7), which
 contradicts with the optimality of  $
\left(P_i^*,P_j^*\right)$ for problem (F.P7).
Therefore, when $\alpha\rightarrow \infty$, the inequality  $P_i^*<P_j^*$ does not hold, i.e.,
 $P_i^*\geq P_j^*$. Based on Proposition \ref{proposition2},
 $P_i^*=P_j^*$ can be obtained.
\vspace{-1em}
 \section{Proof of Lemma \ref{lemma_KKT}}\label{proof_lemma_kkt}
  The Lagrangian function of  problem (R.P2) is first expressed as
  \begin{align}
    \mathcal{L}(\{b_i\},\{\lambda_i\})
    \triangleq \sum_{i=1}^{K} u_{\alpha}(R_i(b_i,b_{i+1}))
    -\sum_{i=1}^{K}\lambda_i(b_{i+1}-b_{i}) \end{align}
  where we define  $\{b_i\}\triangleq\{b_2,\cdots,b_K\}$ and $\{\lambda_i\}\triangleq
  \{\lambda_1,\cdots,\lambda_K\}$, $\lambda_i\geq 0$, are Lagrange multipliers. Based on
  the definition of
  $\mu_{\alpha}(x)$ in \eqref{u(F_k)},
  the KKT conditions are given by
\begin{align}\frac{\partial \mathcal{L}}{\partial b_{k+1}}=&-\frac{\left(R_{k}(b_{k},b_{k+1})\right)^{-\alpha}}{b_{k+1}+\frac{1}{H_k}}
+\frac{\left(R_{k+1}(b_{k+1},b_{k+2})\right)^{-\alpha}}{b_{k+1}
+\frac{1}{H_{k+1}}}\nonumber\\
&-\lambda_k+\lambda_{k+1}=0,\ \forall k\in[1:K-1].\label{largrange_rate_bk1}
\end{align}
 The complementary slackness conditions can be written as
\begin{align}
  \lambda_{k+1}(b_{k+1}-b_{k+2})=0,\\
  \lambda_k(b_{k+1}-b_{k})=0.
\end{align}
Note that $R_{k+1}\left(b_{k+1},b_{k+2}\right)=0$ if $b_{k+1}=b_{k+2}$, and
$R_{k}\left(b_{k},b_{k+1}\right)=0$ if $b_{k}=b_{k+1}$. However,
 from \eqref{largrange_rate_bk1}, $R_{k+1}\left(b_{k+1},b_{k+2}\right),R_{k}\left(b_{k},b_{k+1}\right)>0$
 needs to be satisfied, so
we have $b_{k+2}<b_{k+1}<b_k$ at the optimal solution,
and hence $\lambda_k=\lambda_{k+1}=0$. Thus, from \eqref{largrange_rate_bk1},
\begin{align}-\frac{\left(R_{k}(b_{k},b_{k+1})\right)^{-\alpha}}{b_{k+1}
+\frac{1}{H_k}}&+\frac{\left(R_{k+1}(b_{k+1},b_{k+2})\right)^{-\alpha}}{b_{k+1}
+\frac{1}{H_{k+1}}}=0, \nonumber\\& \forall k\in[1:K-1].\label{largrange_rate_eq}
\end{align}
The above equation can be equivalently transformed to ${f}_{1,k}(b_{k},b_{k+1},b_{k+2})=0$
as defined in \eqref{lemma_bk1_eq2}, which
completes the proof of this lemma.
\vspace{-1em}
\section{Proof of Theorem \ref{theorem_unique}}\label{proof_theorem_unique}


 Denote $(\hat{b}_2,\cdots,\hat{b}_K)$ as a solution of the KKT functions in \eqref{lemma_bk1_eq2}.
Now, we  verify that   $(\hat{b}_2,\cdots,\hat{b}_K)$ is the unique solution of
 these functions. 
To prove this theorem,  {\em reduction to absurdity} is adopted.
  In particular, we assume that, beyond $(\hat{b}_2,\cdots,\hat{b}_K)$,
   there also exists another solution  $(\hat{\hat{b}}_2,\cdots,\hat{\hat{b}}_K)$
  satisfying the KKT conditions in \eqref{lemma_bk1_eq2}.
Assume without loss of generality that $\hat{\hat{b}}_{K}>{\hat{b}}_{K}$.
Let $k=K-1$ in \eqref{lemma_bk1_eq2}, then we have
\begin{align}
 \ln&\left(1+{\hat{b}}_{K-1}H_{K-1}\right)=
  {\ln\left({1+{\hat{b}}_{K}H_{K-1}}\right)}\nonumber\\
  &\qquad+f_{2,K-1}\left({\hat{b}}_{K}\right)\ln\left(1+{\hat{b}}_{K} H_K\right),\\
  \ln&\left({1+\hat{\hat{b}}_{K-1}H_{K-1}}\right)=
 {\ln\left({1+\hat{\hat{b}}_{K}H_{K-1}}\right)}\nonumber\\
 &\qquad+f_{2,K-1}\left(\hat{\hat{b}}_{K}\right)\ln\left(1+\hat{\hat{b}}_{K} H_K\right),
  \end{align}
  where function $f_{2,K-1}(x)$ is defined as
\begin{align}\label{f2x}
f_{2,k}(x)\triangleq \left(\frac{x+\frac{1}{H_{k+1}}}{x+\frac{1}{H_{k}}}
    \right)^{1/\alpha}&= \left(1-\frac{\frac{1}{H_k}-\frac{1}{H_{k+1}}}{x+\frac{1}{H_{k}}}
    \right)^{1/\alpha}, \nonumber\\ &\hspace{0cm}k\in[1:K-1].\end{align}
   Since $f_{2,K-1}(x)$ increases with $x$ when $x>0$, we can
  obtain  \begin{align}\label{two_resluts}&\hat{\hat{b}}_{K-1}>{\hat{b}}_{K-1},\nonumber\\
   \textrm{ and } &\ln\left(\frac{{1+\hat{\hat{b}}_{K-1}H_{K-1}}}
  {{1+\hat{\hat{b}}_{K}H_{K-1}}}\right)>\ln\left(\frac{{1+{\hat{b}}_{K-1}H_{K-1}}}
  {{1+{\hat{b}}_{K}H_{K-1}}}\right).\end{align}

  Now,  let $k=K-2$ in \eqref{lemma_bk1_eq2}, we have
  \begin{align}  \ln&\left({1+{\hat{b}}_{K-2}H_{K-2}}\right)=
 {\ln\left({1+{\hat{b}}_{K-1}H_{K-2}}\right)}\nonumber\\
 &+f_{2,K-2}\left(\hat{b}_{K-1}\right)\ln\left(\frac{1+\hat{b}_{K-1}
  H_{K-1}}{1+\hat{b}_{K} H_{K-1}}
 \right),\label{eq_add1}\\
 \ln&\left({1+\hat{\hat{b}}_{K-2}H_{K-2}}\right)=
 {\ln\left({1+\hat{\hat{b}}_{K-1}H_{K-2}}\right)}\nonumber\\
 &+f_{2,K-2}\left(\hat{\hat{b}}_{K-1}\right)
 \ln\left(\frac{1+
 \hat{\hat{b}}_{K-1} H_{K-1}}{1+\hat{\hat{b}}_{K} H_{K-1}}
 \right).
  \label{eq_add2}\end{align}
  Based on \eqref{two_resluts}, \eqref{eq_add1} and \eqref{eq_add2}, we have
  \begin{align}&\hat{\hat{b}}_{K-2}>{\hat{b}}_{K-2}, \nonumber\\ \textrm{ and } & \ln\left(\frac{{1+\hat{\hat{b}}_{K-2}H_{K-2}}}
  {{1+\hat{\hat{b}}_{K-1}H_{K-2}}}\right)>\ln\left(\frac{{1+{\hat{b}}_{K-2}H_{K-2}}}
  {{1+{\hat{b}}_{K-1}H_{K-2}}}\right).\end{align}

   By analogy,
  $\hat{\hat{b}}_{k}>{\hat{b}}_{k}$ can be verified from $k=K-3$ to $k=1$, i.e., $\forall k\in[1:K-3]$.
   However,  $\hat{\hat{b}}_{K-1}={\hat{b}}_{K-1}=P$ holds
  for problem (R.P2), which contradicts with the result that $\hat{\hat{b}}_{k}>{\hat{b}}_{k}$,
    $\forall k\in[1:K-3]$. Therefore,  only a unique  solution $(\hat{b}_2,\cdots,\hat{b}_K)$ of problem (R.P2) exists to satisfy the   KKT conditions in \eqref{lemma_bk1_eq2}.

\section{Proof of Theorem \ref{theorem_increase}}\label{proof_theorem_increase}
This theorem is proven based on the  {\em inductive method}. Specifically, for a given $t_0\geq 1$,
   we   assume that ${b}_{k}^{(t_0)}>{b}_k^{(t_0-1)}$, $\forall k\in[2:K]$,
    and then we prove that ${b}_{k}^{(t_0+1)}>{b}_{k}^{(t_0)}$, $\forall k\in[2:K]$.
    First, recall that function $f_{2,K-1}(x)$ is defined  \eqref{f2x}, which increases with $x$.
    Next, three different cases are considered.
\vspace{-1em}
\subsection{Case $k=1$}
 In the $t$-th iteration, from Algorithm II and Lemma \ref{lemma_bk1}, we have
 \begin{align}\label{eq_ck}
   {\ln\left(\tilde{c}_{1}^{(t)}\right)}=
    {f_{2,1}\left(b_{2}^{(t)}\right)\ln\left(\tilde{\tilde{c}}_{2}^{(t)}\right)},\ 
  \end{align}
when $k=1$, where we define
\begin{align}\label{c_kt}
\tilde{c}_{k}^{(t)}\triangleq \frac{1+H_{k}b_{k}^{(t)}}{1+H_{k}b_{k+1}^{(t)}},&\
\tilde{\tilde{c}}_{k+1}^{(t)}\triangleq \frac{1+H_{k+1}b_{k+1}^{(t)}}{1+H_{k+1}b_{k+2}^{(t-1)}},\nonumber\\
&\forall k\in[1:K-1].\end{align}

Next, we  consider two cases: $\tilde{\tilde{c}}_{2}^{(t_0+1)}\geq \tilde{\tilde{c}}_{2}^{(t_0)}$ and $\tilde{\tilde{c}}_{2}^{(t_0+1)}< \tilde{\tilde{c}}_{2}^{(t_0)}$.

If
$\tilde{\tilde{c}}_{2}^{(t_0+1)}\geq \tilde{\tilde{c}}_{2}^{(t_0)}$, since we have assumed that $b_{3}^{(t_0)}>b_{3}^{(t_0-1)}$,
obviously $b_{2}^{(t_0+1)}>b_{2}^{(t_0)}$ holds based on \eqref{c_kt}.

If $\tilde{\tilde{c}}_{2}^{(t_0+1)}< \tilde{\tilde{c}}_{2}^{(t_0)}$,
 we adopt {\em reduction to absurdity} to prove that $b_{2}^{(t_0+1)}>b_{2}^{(t_0)}$.
  Specifically, we
 { assume that} $b_{2}^{(t_0+1)}\leq b_{2}^{(t_0)}$,
so we have $f_{2,1}\left(b_{2}^{(t_0+1)}\right)\leq f_{2,1}\left(b_{2}^{(t_0)}\right)$. Thus,
$f_{2,1}\left(b_{2}^{(t_0+1)}\right)\ln(\tilde{\tilde{c}}_{2}^{(t_0+1)})< f_{2,1}\left(b_{2}^{(t_0)}\right)
\ln\left(\tilde{\tilde{c}}_{2}^{(t_0)}\right)$ can be obtained.
From \eqref{eq_ck}, $\ln\left(\tilde{c}_{1}^{(t_0+1)}\right)<\ln\left(\tilde{c}_{1}^{(t_0)}\right)$ holds. However,
from \eqref{c_kt}, $\ln\left(\tilde{c}_{1}^{(t_0+1)}\right)\geq \ln\left(
\tilde{c}_{1}^{(t_0)}\right)$ under the assumption that $b_{2}^{(t_0+1)}\leq b_{2}^{(t_0)}$,
since $b_{1}^{(t_0+1)}=b_{1}^{(t_0)}=P$.
This implies that the assumption $b_{2}^{(t_0+1)}\leq b_{2}^{(t_0)}$ does not hold,
and thus $b_{2}^{(t_0+1)}>b_{2}^{(t_0)}$.
\vspace{-1em}
\subsection{Case $k\in[2:K-2]$}
 Similarly, in the $t$-th iteration, from Algorithm II and Lemma \ref{lemma_bk1}, we have
 \begin{align}\label{eq_ck2}
   {\ln\left(\tilde{c}_{2}^{(t)}\right)}=
    {f_{2,2}\left(b_{3}^{(t)}\right)\ln\left(\tilde{\tilde{c}}_{3}^{(t)}\right)}
   ,\ 
  \end{align}
when $k=2$. As in the previous case, $b_{2}^{(t_0+1)}>b_{2}^{(t_0)}$ if
 $\tilde{\tilde{c}}_{3}^{(t_0+1)}\geq \tilde{\tilde{c}}_{3}^{(t_0)}$.

Now, {\em reduction to absurdity} is also adopted if $\tilde{\tilde{c}}_{3}^{(t_0+1)}< \tilde{\tilde{c}}_{3}^{(t_0)}$.
 Specifically, similar to the previous case $k-1$,
 $\ln\left(\tilde{c}_{2}^{(t_0+1)}\right)<\ln\left(\tilde{c}_{2}^{(t_0)}\right)$ can be obtained,
  if we assume that $b_{2}^{(t_0+1)}\leq b_{2}^{(t_0)}$ in \eqref{eq_ck2}.
   However,
from  \eqref{c_kt}, $\ln\left(\tilde{c}_{2}^{(t_0+1)}\right)\geq \ln\left(
\tilde{c}_{2}^{(t_0)}\right)$  under the assumption that $b_{3}^{(t_0+1)}\leq b_{3}^{(t_0)}$,
since $b_{2}^{(t_0+1)}>b_{2}^{(t_0)}$ as verified in the previous case.
This implies that the assumption $b_{3}^{(t_0+1)}\leq b_{3}^{(t_0)}$ does not hold,
and  thus $b_{3}^{(t_0+1)}>b_{3}^{(t_0)}$.

Similarly, $b_{k}^{(t_0+1)}>b_{k}^{(t_0)}$ can be proven iteratively, for $k\in[3:K-2]$.
\vspace{-1em}
\subsection{Case $k=K-1$}
From Algorithm II and Lemma \ref{lemma_bk1}, we have
 \begin{align}\label{eq_ck3}
  {\ln\left(\tilde{c}_{K-1}^{(t)}\right)}=
    {f_{2,K-1}\left(b_{K}^{(t)}\right)\ln\left(\tilde{\tilde{c}}_{K}^{(t)}\right)}
   ,\ 
  \end{align}
when $k=K-1$.  Note that $b_{K}^{(t_0+1)}>b_{K}^{(t_0)}$ can be proven using almost the same
steps to the previous two cases. There is only a slight difference that is $b_{K+1}^{t_0}
=b_{K+1}^{(t_0-1)}=0$. In order to show that $b_{K}^{(t_0+1)}>b_{K}^{(t_0)}$ if
$\tilde{\tilde{c}}_{2}^{(t_0+1)}\geq \tilde{\tilde{c}}_{2}^{(t_0)}$,
 we only need to verify that  $\tilde{\tilde{c}}_{2}^{(t_0+1)}\neq \tilde{\tilde{c}}_{2}^{(t_0)}$.
 Specifically, from \eqref{c_kt}, $b_{K}^{(t_0+1)}
=b_{K}^{(t_0)}$ if $\tilde{\tilde{c}}_{2}^{(t_0+1)}= \tilde{\tilde{c}}_{2}^{(t_0)}$,
so ${\tilde{c}}_{2}^{(t_0+1)}={\tilde{c}}_{2}^{(t_0)}$ from \eqref{eq_ck3} and
$b_{K-1}^{(t_0)}=b_{K-1}^{(t_0-1)}$ can be obtained from \eqref{c_kt}.
 However, $b_{K-1}^{(t_0+1)}>b_{K-1}^{(t_0)}$
as verified in the previous case, which means that $\tilde{\tilde{c}}_{2}^{(t_0+1)}= \tilde{\tilde{c}}_{2}^{(t_0)}$ does not hold, i.e., $\tilde{\tilde{c}}_{2}^{(t_0+1)}\neq \tilde{\tilde{c}}_{2}^{(t_0)}$.

\bibliographystyle{ieeetr}
\bibliography{references}

\begin{IEEEbiography}[{\includegraphics[width=1.1in,
height=1.35in,clip, keepaspectratio]{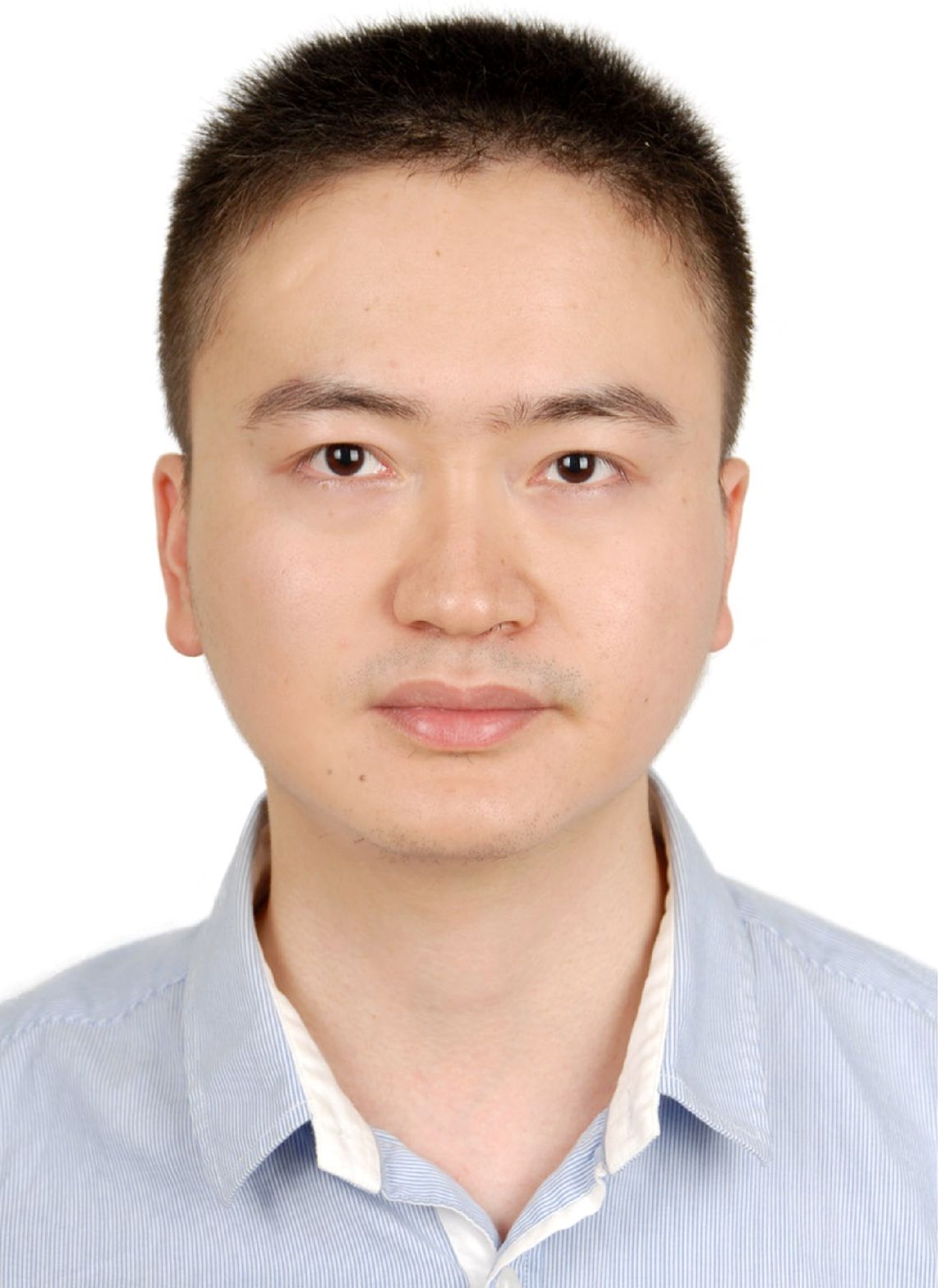}}]{Peng
Xu} received the B.Eng. and the Ph.D. degrees in electronic and information engineering from the University of Science and Technology of China, Anhui, China, in 2009 and 2014, respectively. Since July 2014, he was working as a postdoctoral researchers
 with the Department of Electronic Engineering and Information Science, University of Science and Technology of China, Hefei, China.
 From July 2016, he has been working  at
 the School of Communication and Information Engineering, Chongqing University of Posts and
Telecommunications (CQUPT), Chongqing, China.
 His current research interests include cooperative communications, information theory,  information-theoretic secrecy,
and 5G networks. Dr. Peng Xu received IEEE Wireless Communications Letters Exemplary Reviewer 2015.\end{IEEEbiography}\vspace{-1em}

\begin{IEEEbiography}[{\includegraphics[width=1in,
height=1.25in,clip, keepaspectratio]{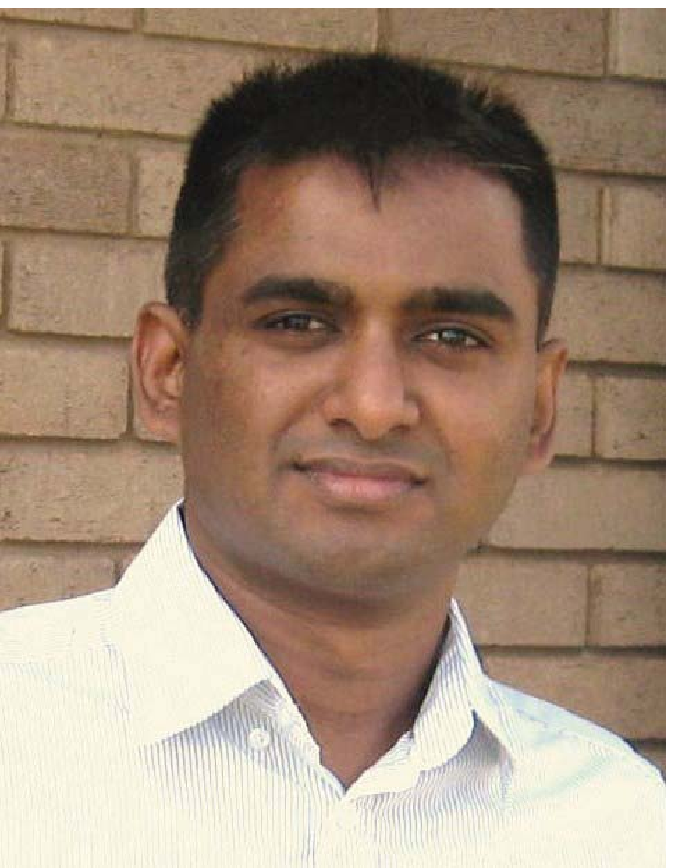}}]{Kanapathippillai Cumanan} (M'10)
received the BSc degree with first class honors in electrical and electronic engineering from the University of Peradeniya, Sri Lanka in 2006 and the PhD degree in signal processing for wireless communications from Loughborough University, Loughborough, UK, in 2009.

He is currently a lecturer at the Department of Electronic Engineering, University of York, UK. From March 2012 to November 2014, he was working as a research associate at School of Electrical and Electronic Engineering, Newcastle University, UK. Prior to this, he was with the School of Electronic, Electrical and System Engineering, Loughborough University, UK. In 2011, he was an academic visitor at Department of Electrical and Computer Engineering, National University of Singapore, Singapore. From January 2006 to August 2006, he was a teaching assistant with Department of Electrical and Electronic Engineering, University of Peradeniya, Sri Lanka. His research interests include physical layer security, cognitive radio networks, relay networks, convex optimization techniques and resource allocation techniques.

Dr. Cumanan was the recipient of an overseas research student award scheme (ORSAS) from Cardiff University, Wales, UK, where he was a research student between September 2006 and July 2007.

  \end{IEEEbiography}\vspace{-1em}

\end{document}